\documentclass[smallextended]{svjour3}       
\smartqed  
\usepackage[ruled]{algorithm2e}
\pdfoutput=1
\renewcommand{\algorithmcfname}{ALGORITHM}
\SetAlFnt{\small}
\SetAlCapFnt{\small}
\SetAlCapNameFnt{\small}
\SetAlCapHSkip{0pt}
\IncMargin{-\parindent}
\usepackage{paralist}
\usepackage{amsfonts}
\usepackage{amsmath,amssymb}
\usepackage{graphicx}
\usepackage{color}
\usepackage{amssymb}
\usepackage{mathrsfs}
\usepackage{stmaryrd}
\usepackage{graphicx}

\newcommand{\R}{\mathbb{R}}

\newcommand{\N}{\mathbb{N}}

\newcommand{\Q}{\mathbb{Q}}

\newcommand{\Z}{\mathbb{Z}}

\newcommand{\ee}{\varepsilon}

\newcommand{\Sc}{\mathcal{S}}
\newcommand{\Gc}{\mathcal{G}}

\newcommand{\FF}{\mathcal{F}}
\newcommand{\Ec}{\mathcal{E}}
\newcommand{\Vc}{\mathcal{V}}

\newcommand{\Pc}{\mathcal{P}}

\newcommand{\Ic}{\mathcal{I}}
\newcommand{\Hc}{\mathcal{H}}
\newcommand{\Lc}{\mathcal{L}}

\newcommand{\Cc}{\mathcal{C}}
\newcommand{\Oc}{\mathcal{O}}
\newcommand{\head}[1]{{#1}^-}
\newcommand{\tail}[1]{{#1}^+}

\newcommand{\lo}{\longrightarrow}
\newcommand{\li}{\left}
\newcommand{\re}{\right}
\newcommand{\mi}{\,\,\big|\,\,}
\newcommand{\eq}{\Longleftrightarrow}

\newcommand{\Os}{\mathrm{O}_{\mathrm{si}}}
\newcommand{\Oe}{\mathrm{O}_{\mathrm{el}}}
\newcommand{\Ps}{\mathrm{P}_{\mathrm{si}}}
\newcommand{\Pe}{\mathrm{P}_{\mathrm{el}}}
\newcommand{\el}{\mathrm{el}}
\newcommand{\si}{\mathrm{si}}
\newcommand{\spann}{\mathrm{span}}

\newcommand{\mult}{\mathrm{mult}}

\newcommand{\argmin}{\arg\!\min}
\newcommand{\argmax}{\arg\!\max}

 \newtheorem{lm}{Lemma}[section]
 \newtheorem{theo}{Theorem}

\newtheorem{pro}{Proposition}[section]
\newtheorem{strat}{Strategy}

\journalname{Preprint version submitted to Theory of Computing Systems}
\begin{document}

\title{Exact Localisations of Feedback Sets  }


\author{Michael Hecht}


\institute{Michael Hecht \at
Now at : MOSAIC Group, Chair of Scientific Computing for Systems Biology, Faculty of Computer Science, TU Dresden
\& Center for Systems Biology Dresden, Max Planck Institute of Molecular Cell Biology and Genetics, Dresden  \\
\email{hecht@mpi-cbg.de}            \\
$ $\\
This work was partially funded by the DFG project ``MI439/14-1''.
}


\maketitle

\begin{abstract}
The \emph{feedback arc (vertex) set problem}, shortened  FASP (FVSP), is to transform a given multi
digraph $G=(V,E)$ into an acyclic graph by deleting as few arcs (vertices) as
possible. Due to the results of Richard M. Karp in 1972 it is one of the
classic NP-complete problems. An important contribution of this paper is that the subgraphs
$G_{\el}(e)$, $G_{\si}(e)$ of all \emph{elementary cycles} or \emph{simple  cycles} running through some arc $e \in E$,  
can be computed in $\mathcal{O}\big(|E|^2\big)$ and $\mathcal{O}(|E|^4)$, respectively. We use this fact and introduce the notion of the \emph{essential minor} and \emph{isolated cycles}, 
which yield a priori problem size reductions and  in the special case of so called \emph{resolvable graphs} an exact solution in 
$\mathcal{O}(|V||E|^3)$. 
We show 
that weighted versions of the FASP and FVSP possess a Bellman decomposition, which yields exact solutions 
using a dynamic programming technique in times 
$\mathcal{O}\big(2^{m}|E|^4\log(|V|)\big)$ and $\mathcal{O}\big(2^{n}\Delta(G)^4|V|^4\log(|E|)\big)$, 
where  $m \leq  |E|-|V| +1$,  $n \leq (\Delta(G)-1)|V|-|E| +1$, respectively. The parameters $m,n$ can be computed in $\mathcal{O}(|E|^3)$, $\mathcal{O}(\Delta(G)^3|V|^3)$, respectively
and denote the maximal 
dimension of the cycle space of all appearing
\emph{meta graphs}, decoding the intersection behavior of the cycles. Consequently, $m,n$  equal zero if all meta graphs are 
trees. Moreover, we deliver several heuristics  and discuss how to control their variation from the optimum. 
Summarizing, the  presented results allow us to suggest a strategy for an implementation of a fast and accurate FASP/FVSP-SOLVER. 
\end{abstract}

\section{Introduction}
\label{sec:intro}

The \emph{feedback arc set} problem, shortened FASP,  is to delete as less as possible arcs of a graph such that 
the resulting subgraph is acyclic, i.e., it contains no directed cycle. 
Another equivalent formulation is to find a linear ordering of the vertices of
the graph such that the number of back arcs is minimized.  Therefore the
problem is also known as \emph{maximum acyclic subgraph} problem or
\emph{linear ordering problem}.
For directed graphs this problem is one of the classic NP-complete problems
\cite{Karp:1972}. 
The problem of deleting  a smallest subset of vertices to result in an acyclic subgraph is known as \emph{feedback vertex set problem} (FVSP).  The
FASP and FVSP are linear time reducible among each other, by keeping the
relevant parameters fix as we will assert in section \ref{ssec:notation}, alternatively  see \cite{Even:1998}. Therefore 
algorithmic properties of the two problems are transferable. In particular, 
the FVSP is also NP-complete. 
Analogous problems for undirected graphs can be defined. As shown in \cite{Karp:1972} the feedback
vertex set problem remains NP-complete, while the feedback arc set problem can be solved
efficiently by solving a maximum spanning tree problem. 
An excellent overview on feedback sets can be found in \cite{bang}. 
The problem of finding minimal transversals of directed cuts is closely related 
to the FASP, see \cite{Lucchesi:1978}. The FASP stays NP-complete for graphs where
every node has an in-degree and out-degree of at most three or line digraphs
even when every clique has at most size three \cite{Gavril:1977}. 
It is also NP-complete for tournament graphs \cite{Alon:2006}. 
However, there also exist graph classes possessing polynomial time algorithms, e.g., planar directed graphs or more general
weakly acyclic digraphs \cite{Groetschel:1985}, and 
reducible flow graphs \cite{Ramachandran:1988}. 
The FASP or FVSP has a multitude of applications, e.g., retiming synchronous circuitry  \cite{leiserson},  circuit testing \cite{kunzmann},  computational biology and
neuroscience \cite{greedy}, network analysis and operating systems  \cite{Silber}. 

\subsection{Outline} In section \ref{Prel} we provide the graph theoretical concepts, which are fundamental for this article. In section \ref{main} we present our main results. The fact that 
the FASP/FVSP on multi-digraphs can be reduced to simple graphs is asserted in section \ref{sec:essential} and the first  a priori problem size reduction is deduced. 
In section \ref{sec:elecyc} we construct an algorithm determing the induced subgraph of all cycles with one arc in common efficently. This knowledge is used in section \ref{sec:iso} to introduce the concept of 
\emph{isolated cycles} and \emph{resolvable graphs} and providing an efficent solution of the FASP/FVSP on resolvable graphs. Moreover, the second a priori problem size reduction is given.  Afterwards we concentrate on the main 
result of the article. Namely, that the FASP/FVSP possesses a Bellmann decomposition and present exact solutions using this fact to apply a dynamic programming technique in section \ref{Bell}. 
In section \ref{Greedy}, we discuss how greedy approaches behave with respect to the problems and develope a strategy for a general FASP/FVSP-SOLVER. Finally, we discuss our results and other alternatives in 
section \ref{Algo}.

\section{Preliminaries}\label{Prel}

Before we can introduce the FASP (FVSP) formally, some basic concepts of graphs and 
cycles need to be mentioned.  

\subsection{Graphs and Cycles}
\label{ssec:notation}
A \emph{multi-directed graph}, or \emph{multi-digraph}, $G=(V,E)$ consists of a set of vertices
$V$ and a multi-set of arcs $E$ containing elements from $V\times V$. A \emph{directed graph} or \emph{digraph} is a multi-digraph with a 
simple arc set $E$, i.e., $E \subseteq V\times V$ and therefore every $e \in E$ occurs exactly once. 
A digraph is called  \emph{simple} if there are no loops, i.e., 
$E\cap \mathbb{D}(V\times V)= \emptyset$, where $\mathbb{D}(V\times V)=\{(v,v) \in V\times  V\mi v\in V\}$ denotes the diagonal.

If not stated otherwise throughout the 
article $G =(V,E)$ denotes a finite, connected, directed and
loop-free multi-digraph and $G\setminus e$, $G\setminus v$ denote the graphs which occur by deleting the arc $e$ 
and possibly occurring isolated vertices or the vertex $v$ and all its adjacent arcs. For $\ee \subseteq E$ and $\nu \subseteq V$ the graphs $G\setminus \ee$, $G\setminus \nu$ are analogously defined. Moreover, 
 $\Gc(\cdot)$, $\Ec(\cdot)$, $\Vc(\cdot)$ denote the induced graph, the set of all arcs, the set of all vertices 
which are inherited by a set or set system of graphs, arcs or vertices.
With $\Pc(A)$ we denote the power set of a given set $A$. 

For an arc $e=(u,v) \in E$ we denote $\tail{e}=u$ as the \emph{tail} and $\head{e}=v$ as the
\emph{head}  of the arc.  Two arcs $e$ and $f$ are called \emph{consecutive} if
$\head{e}=\tail{f}$ and are called \emph{connected} if $\{\head{e},\tail{e}\}\cap
\{\head{f},\tail{f}\} \not =\emptyset$. 
A \emph{directed path} from a vertex $u$ to a vertex $v$ is a sequence of
consecutive arcs where $u$ is the tail of the first arc and $v$ is the head of
the last. A \emph{connected path} from a vertex $u$ to a vertex $v$ is a
sequence of connected arcs containing $u$ and $v$ as vertices. 
A digraph is \emph{connected} if there is a connected path between every pair of
its vertices. A \emph{weighted digraph} $(G,\omega)$ or $(V,E,\omega)$, is a
digraph with  an additional weight function $\omega:E\lo \R$, which assigns a
(usually positive) weight to each arc. 
For a given vertex $v \in V$ the sets $N^{+}_V(v):=\li\{ u \in V \mi (v,u)\in E\re\}\,, N^{-}_V(v):=\li\{ u \in V \mi (u,v)\in E\re\}$,  
$N^{\pm}_E(v):=\li\{ e \in E \mi e^{\pm} = v\re\} $
shall denote the set of all outgoing or incoming vertices or  arcs of $v$ respectively. The \emph{indegree} (respectively \emph{out
degree}) of a vertex $u$ is given by $\deg^{\pm}(u) = |N^{\pm}_E(u)|$ and the degree of a vertex is 
$\deg(u)=\deg^-(u)+\deg^+(u)$. $\Delta^\pm(G), \Delta(G)$ shall denote the maximal (in/out) degree, respectively. 

A directed (connected) \emph{cycle} of a digraph is a multiset of arcs $\{e_0,\dots,e_k\}$ such that  
there is a permutation  $\phi: \{0,\dots,k\} \lo \{0,\dots,k\}$ with $e_{\phi(i)}$ 
and $e_{\phi(i)+1 \mod k+1}$ are consecutive (connected), for $1\leq i\leq k$.  
A \emph{loop} is a cycle containing only a single arc. A cycle is \emph{simple}
if the set of contained arcs $\{e_1,\dots,e_k\}$ is a simple set, i.e., it
visits every arc, it contains, exactly once.  A cycle is \emph{elementary} if
each vertex it contains is visited exactly once. We denote with $O(G)$ the set of all directed cycles and with $\Oe(G)$, $\Os(G)$ the set of all elementary or simple cycles, respectively. 
Analogously, $\Oe^0(G)$ and $\Os^0(G)$ shall denote the set of all \emph{connected} (and not necessarily directed) \emph{elementary} and \emph{simple cycles} respectively. If not stated otherwise in the whole article  a cycle is assumed to be directed and elementary.
A \emph{feedback vertex set} (FVS) of $G$ is a set $\nu \in \Pc(V)$ such that $G\setminus \nu$ is \emph{acyclic}, i.e., $G \setminus \nu$ 
contains no directed cycle.  A \emph{feedback arc set} (FAS) of $G$ is a set $\ee \in \Pc(E)$ such that $G\setminus \ee$ is acyclic.

\begin{definition}[line graph, natural hypergraph] The directed  \emph{line graph} $L(G)=(V_L,E_L)$ of a digraph $G$ is a digraph where each vertex
represents one of the arcs of $G$ and two vertices are connected by an arc
if and only if the corresponding arcs are consecutive.
In contrast the \emph{natural hypergraph} $\Hc(G)=(\bar V,\bar E)$ of $G$ is constructed by identifying the arcs of $G$ 
with the vertices of $\Hc(G)$, i.e., $\bar V$ is a simple set of vertices such that $|V| = |E|$, where $|E|$ is counted with multiplicities. By fixing the identification $\bar V \cong E$  we introduce a directed  
hyperarc $h_v$ for every vertex $v \in V$ by requiring that head and tail coincide 
with  all outgoing and ingoing arcs respectively, i.e., 
$h_v = \big(N_E^-(v), N_E^+(v)\big)$. Consequently, $\bar E \cong V$ and therefore every hyperarc can be labeled by its corresponding vertex.
See Figure \ref{fig:hyp}  for an example. 
\end{definition}
The directed, elementary cycles of the line graph $L(G)$ of $G$ are in 1 to 1 correspondence to the directed, simple cycles of $G$ 
while the directed simple cycles of $\Hc(G)$, i.e., directed cycles which run through every hyper arc
exactly once are called \emph{Berge cycles}, \cite{bergehyp}. Note that if $G$ is a simple digraph, i.e., there are no multi arcs, then 
the set of all Berge cycles of $\Hc(G)$ are in 1 to 1 correspondence to the set of all elementary cycles of $G$.  For a set of hyperarcs $\bar \ee \subseteq \bar E$
of $\Hc(G)$ we denote with $\ee \subseteq V$ the corresponding vertices in $G$. 
We summarize  some facts in this regard.

\begin{proposition} Let $G=(V,E)$ be a graph. 
\begin{enumerate}
 \item[i)] The line graph $L(G)=(V_L,E_L)$  can be constructed in $\Oc\big(|E|^2\big)$. 
 \item[ii)] The natural hypergraph $\Hc(G)=(\bar V,\bar E)$ can be constructed in $\Oc\big(\Delta(G)|V|\big)$. 
 \item[iii)]$\nu \subseteq V_L$ is a FVS of $L(G)$ if and only if $\nu$ is a FAS of $G$. 
 \item[iv)] $\bar \ee \subseteq \bar E$ is a FAS of $\Hc(G)$, with respect to the notion of Berge cycles, if and only if $\ee$ is a FVS of $G$.  
\end{enumerate}
\label{berge}
\end{proposition}
\begin{proof} Storing $G$ as a adjacency list and following  the definitions immediately implies $i)$ and $ii)$. Since any two vertices $e,f$ of the line graph $L(G)$ 
are adjacent if and only if the corresponding arcs are consecutive in $G$ and any two arcs $h_u,h_v$ of $\Hc(G)$ are consecutive, i.e., 
$h_u^- \cap h_v^+\not = \emptyset$, if and only if the corresponding vertices $u$ and $v$ are adjacent in $G$, $ii)$ and $iii)$ follow. 
\qed 
\end{proof}

\begin{remark}\label{BIP} Note that by introducing an additional arc $h^*_v$ between head and tail of every hyperarc $h_v=(N_E^-(v),N_E^+(v))$ of $\Hc(G)$, the natural hypergraph becomes a directed graph 
$G^*=(V^*,E^*)$ with $|V^*|=|E|+2|V|$, $|E^*|\leq \Delta(G)|V|+|V|$. The directed cycles of $G^*$ are in 
1 to 1 correspondence to the Berge cycles of $\Hc(G)$ and a FAS $\ee$ of $G^*$ is in  1 to 1 correspondence to a FAS $\bar \ee$ of $\Hc(G)$ by identifying $\bar \ee$ with the additional introduced arcs belonging to the 
hyperarcs in $\bar \ee$ and identifying $\ee$ with the hyperarcs corresponding to the bipartite graphs cutted by $\ee$. If $\gamma : \bar E \lo \R+$ is an arc weight on $\Hc(G)$ then setting 
$\gamma^*(h) \equiv \gamma(h_v)$ for all $h \in N_E^-(v) \cup N_E^+(v)\cup\{h^*_v\}$ yields the translated weight.

\end{remark}

\begin{figure}[t!]
 \centering
 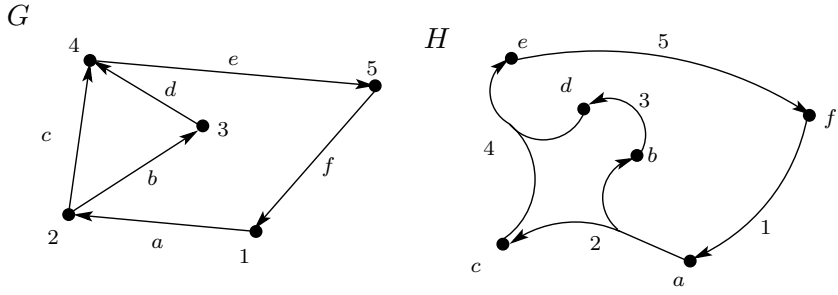
 \caption{The natural hyper graph $\Hc(G)$ of $G$. }
 \label{fig:hyp}
\end{figure}
To give a more algebraic notion of cycles we consider
$$ X(G):= \bigoplus_{e \in E} \Z e $$ 
the  free  $\Z$-module generated by $E$. If we choose coordinates, i.e., a numbering for $E$ and $V$ then we can identify $E$ with $\{e_1,\dots,e_{|E|}\}$, $V$ with $\{v_1,\dots,v_{|V|}\}$ and $X$ with $\Z^{|E|}$. In this case 
an element $x \in X(G)$ is a tuple $x=(x_1,\dots,x_{|E|})$, which can be interpreted as a set of paths through $G$ where $x_i \in \Z$ indicates how often we pass the arc $e_i$ and the sign of $x_i$ determines 
in which direction this is done.
We denote with  
$\Ic(G)$  the incidence matrix of $G$ with respect to these identifications, i.e., $\Ic(G)=(\iota_{ij})_{\genfrac{}{}{0pt}{}{1 \leq i \leq | V|}{1 \leq j \leq |E |}}$ with 
$$ \iota_{ij} = \li\{\begin{array}{rl}
	0  & \,, \quad \text{if} \quad   \tail{e_i}\not= v_j  \,\, \text{and}\,\, \head{e_i}\not = v_j \\
	1  &  \,, \quad \text{if} \quad  \tail{e_i}= v_j  \,\, \text{and}\,\, \head{e_i}\not = v_j \\
	-1 &  \,, \quad \text{if} \quad  \tail{e_i}\not = v_j \,\, \text{and}\,\, \head{e_i}= v_j
\end{array} \re. \,.$$
It is a well known fact, see for instance \cite{M}, that $x \in X(G)$ is a cycle of $G$ if and only if $\Ic x  = 0$, i.e., 
the submodule of all cycles of $G$ coincides with the set of homogeneous solutions $\Lambda(G)=\ker \Ic(G)$. In particular,  this implies that 
$\Lambda(G)$ is a free $\Z$-module with  
\begin{equation}\label{dim}
 \dim_\Z \Lambda(G) = \dim_\Z (\ker \Ic(G)) = |E| -|V| + \#G,
\end{equation}
where $\#G$ denotes the number of connected components  of $G$ and therefore equals $1$ by our assumption on $G$.
Note that the vector space $\Lambda^0(G):=\Lambda(G)/\Z_2^{|E|}$ can be understood as the cycle space of connected cycles, given by the kernel of the incidence matrix $\Ic^0(G)$ defined with 
respect to $\Z_2$ coefficients. In this case  
\begin{equation}\label{dim2}
 \dim_{\Z_2} \Lambda^0(G) =  \dim_{\Z_2}\ker(\Ic^0(G)) = |E| -|V| + \#G,
\end{equation} 
still holds, see again \cite{M}. 
 
\begin{remark}\label{rem:OG}
If $c \in \Oe(G)$ then $1$ is the only non vanishing entry of $c$, i.e., $c \in
\{0,1\}^{|E|}$.  Moreover, no elementary cycle  is subset of another elementary cycle. Indeed assume the opposite and consider two elementary cycles $c,c^\prime \in \Oe(G)$ with $c\subseteq c^\prime$,
then $\Ic(G)(c^\prime -c)= \Ic(G)c^\prime-\Ic(G)c =0-0=0$. Thus, $c^\prime -c \not = 0$
is also a positive oriented cycle and therefore we have $c^\prime = (c^\prime -c) + c \not \in
\Oe(G)$.  A contradiction!
\end{remark}
Not that every simple cycle is given by the union of arc disjoint elementary cycles. The following example illustrates this fact. 
\begin{example} 
\label{exa:elementary}
Consider the graph $G$ in Figure~\ref{fig:elementary}. Then one observes that the cycle $c = \{e,f,g,h,i\}$ is a simple, non-elementary cycle while the cycles $\{e,f,i\}$,  $\{f,h\}$ and $\{h,g\}$ are 
elementary cycles. 
\end{example}

\begin{figure}[t!]
 \centering
 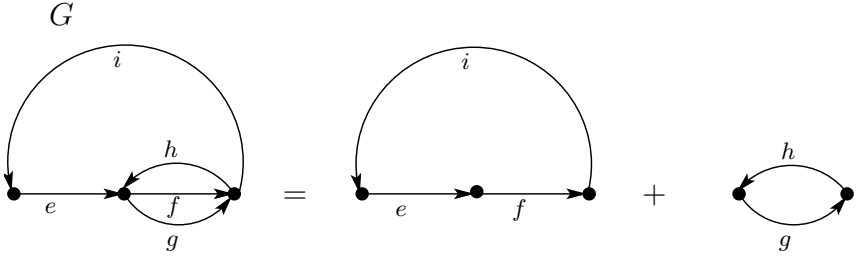
 \caption{Elementary and non-elementary cycles in $G$, see Example~\ref{exa:elementary}.}
 \label{fig:elementary}
\end{figure}

\subsection{The Feedback Arc Set Problem (FASP)} \label{FASP}

Now we have all ingredients to give a formal definition of the FASP. 
\begin{problem}\label{def:fasp}
Let $G =(V,E)$ be  a finite, connected, directed, and loop-free graph, $\omega
: E\lo \N^+$ be a weight function.  Then the \emph{weighted FASP} is to
find a  set of arcs $ \ee \in \Pc(E)$  such that $G \setminus \ee$ is acyclic,
i.e., $\Oe\big(G \setminus \ee\big) = \emptyset$ and  
$$ \Omega_{G,\omega}(\ee):=\sum_{e \in \ee} \omega(e) $$
is minimized. We denote the set of solutions of this problem with $\Sc(G,\omega)$
and denote with $\Omega(G,\omega):=\Omega_{G,\omega}(\ee)$, $\ee \in \Sc(G,\omega)$ the optimal weight or \emph{feedback length}.  
If $\omega$ is constant, e.g., equal to $1$, then the problem coincides with
the \emph{unweighted minimal FASP}. 
\end{problem}

\begin{remark}
The condition on $G$ to be loop-free is not an essential restriction.
This is because every loop is contained in any solution of the minimal FASP.
\end{remark}
\begin{remark}
Note that, every cycle $c \in  O(G)$ can be generated by elementary cycles
$c_1,\dots,c_n \in \Oe(G)$ using only non-negative coefficients. Thus, if $\ee$ is
a solution of Problem \ref{def:fasp} then certainly
$O\big(G\setminus \ee\big) = \emptyset$, which implies that our
notion of acyclic graphs is consistent for the problem. 
\end{remark}

\begin{remark}
Note that for given graph $G=(V,E)$ the smallest
subgraph $G_o\subseteq G$, which contains all cycles of $G$, i.e., $G^\prime= G_o$
whenever $G^\prime \subseteq G_o$ and $\Oe(G)=\Oe(G^\prime)$ can be constructed in $\Oc(|E|^2)$.  Indeed the arc set of $G_o$ is
constructed from $E$ by removing arcs $(u,v)$ if there is no directed path
from $v$ to $u$. For every arc this can be done by depth-first search in linear time if $G$ is stored in an adjacency list.   
Since removing arcs does not generate new cycles it suffices to check every arc once yielding the estimated runtime performance. Certainly, a solution for $G_o$ is a
solution for $G$. We shortly denote with $\Gc_o(G):=G_o$ and with $\Gc_o^0(G):=G^0_o$ the analogous graph appearing by considering connected cycles instead of directed ones.
In particular, by the argumentation above, an elementary or directed cycle $c \in \Oe(G)$, $c' \in \Os(G)$ can be found in $\Oc(|E|^2)$ or no cycle exist. 

\label{GO}
\end{remark}

\begin{remark}
\label{half} 
If $G=(V,E)$ is a simple graph then denoting with $\lfloor\cdot\rfloor$ the Gauss bracket we observe that at most $\lfloor|E|/2\rfloor$ arcs have to 
be deleted to obtain a graph where no connected path of length $2$ exists anymore. 
In particular, the graph is acyclic in this case and therefore 
$$\Omega(G,\omega) \leq \max_{e\in E}\omega(e) \cdot |E|/2.$$  
See also \cite{Saab}.
\end{remark}

\subsection{The Feedback Vertex Problem (FVSP)} \label{FVSP}

Let $G=(V,E)$ be given and $\gamma : V \lo \R^+ $ be a vertex weight.
The feedback vertex set problem (FVSP) on $(G,\gamma)$ is obtained by replacing the role of arcs by
vertices in Problem~\ref{def:fasp}. In regard of Proposition \ref{berge}, by introducing the hyperarc weight $\bar \omega (h_v):=\gamma(v)$, we realize that the FVSP is equivalent to the the FASP on the natural hypergraph 
$\Hc(G)$ of $G$, with respect to the notion of Berge cycles. Already in  Remark \ref{BIP} we mentioned how to treat this case. Vice versa the FASP on an arc weighted graph $(G,\omega)$ 
is equivalent to the FVSP on the line graph
$L(G)$ of $G$ by introducing the vertex weight $\gamma(v) = \omega(v)$, $v \in V_L=E$. Since the described transformations can be done efficently, see Proposition \ref{berge}, 
an efficent solution of the
FASP or FVSP for an arc and vertex weighted instance $(G,\omega,\gamma)$ yields an efficent solution of the FVSP or FASP for the
transformed instances and vice versa. In particular, by summarizing some already known results we obtain: 

\begin{theo} The unweighted FASP and FVSP are APX complete. 
\end{theo}
\begin{proof} Since there is an $L$-reduction of the Minimum Vertex Cover Problem, which is APX complete due to  \cite{Dinur}, to the FVSP, see \cite{Karp:1972},  the FVSP is APX complete.
Proposition \ref{berge} shows that the unweighted FASP on $G$  is equivalent to the unweighted FVSP on $L(G)$. Thus, 
the feedback length of any solution remains unchanged yielding an $L$-reduction of the FASP to the FVSP implying the claim.  
 \end{proof}
We expect that the theorem above still holds for the weighted versions. However, due to the observations made so far, we will focus our studies on the FASP to increase the understanding of 
the localisation of feedback sets.

\section{Main Results} \label{main}

Though the FASP and FVSP are equivalent problems in graph theory and computer sciences the only exact solutions of the FVSP known to us is are  the algorithms of \cite{Razgon} with complexity $\Oc(1,9977^{|V|}|V|^{\Oc(1)})$ and  
\cite{Chen} requiring  $\Oc\big(|E|^44^\Omega\Omega^3\Omega!\big)$, where $\Omega$ denotes the feedback length. 
A detailed comparison to our approach is given in Section \ref{Algo}. For now we just mention our results: 

\begin{theorem}\label{AA} Let $(G, \omega)$ be a graph. Then there is an algorithm testing whether $(G,\omega)$ is resolvable and determing a solution of the  weighted FASP on $G$ in case of resolvability in 
 $\Oc(|V||E|^3)$. 
\end{theorem}
Though there are infinitely many resolvable graphs not all graphs are resolvable. 
However, if the graph $(G, \omega)$ is not resolvable, we still can find an exact solution: 

\begin{theorem}\label{B} Let  $(G,\omega)$ be a graph then there is a parameter  $m \in  \N$, $m \leq \dim_{\Z_2}\Lambda^0(G)= |E|-|V| +1$, which can be determined in $\Oc(|E|^3)$ and an algorithm \emph{CUT} with run time 
$\Oc\big(2^{m}|E|^4\log(|V|)\big)$ solving the weighted FASP.
\end{theorem}

Due to Proposition \ref{berge}  the analogous results with respect to the FVSP hold. In particular,  we call a vertex weighted graph $(G,\nu)$ resolvable iff its natural
hypergraph is resolvable, see section \ref{FVSP} again. If we replace every hyperarc of  $\Hc(G)=(\bar V,\bar E)$ with its corresponding bipartite graph then by following Remark \ref{BIP} we have 
$|V^*|=|E|+2|V|$, $|E^*|\leq (\Delta(G)+1)|V|$ for the resulting graph $G^*=(V^*,E^*)$.
The translated results therefore become:   

\begin{theorem}Let $(G, \nu)$, $\nu :V \lo \R^+$  be a vertex weighted graph. Then there is an algorithm testing whether $(G,\nu)$ is resolvable and determing a solution of the  weighted FVSP on $G$ 
in case of resolvability in  $\Oc(\Delta(G)^3|V|^3|E|)$. 
\end{theorem}

In case of non-resolvability we have: 

\begin{theorem} Let  $(G,\nu)$ be a graph then there is a parameter  $m \in  \N$, $m \leq (\Delta(G)-1)|V|-|E| +1$, which can be determined in $\Oc(\Delta(G)^3|V|^3)$ and an algorithm \emph{CUT} with run time 
$\Oc\big(2^{m}\Delta(G)^4|V|^4\log(|E|)\big)$ solving the weighted FVSP.
\end{theorem}
Note that there are infinitely many instances where $m=0$ and $m \leq |V|$ on homogenous graphs, i.e., if $\Delta(G)\leq (|E|-1)/|V|+2$.  Thus, by computing the bound or directly $m$ 
we can decide whether the algorithm of \cite{Razgon} or our approach will be faster for a given instance and choose the better alternative. 
Moreover, the feedback length $\Omega$ can not assumed to be constantly bounded. Thus,  the algorithm of \cite {Chen} actually possesses a complexity of $\Oc\big(4^{|V|}|V|^3|E|^4|V|!\big)$, 
which is much slower than our approach.
Finally, we want to mention that all theorems are based on the following crucial fact:

\begin{theorem}
 Let $G=(V,E)$ be a graph and $e \in E$. Then there exist algorithms which compute:
 \begin{enumerate}
 \item[i)] The subgraph $G_{\el}(e) \subseteq G$ induced by all elementary  cycles $c \in \Oe(G)$ with $e\in \Ec(c)$  in $\Oc(|E|^2)$. 
 \item[ii)] The subgraph $G_{\si}(e)\subseteq G$ induced by all simple cycles $c \in \Os(G)$ with $e\in \Ec(c)$  in $\Oc(|E|^4)$. 
\end{enumerate}
\end{theorem}
A proof and a more precise version of the statement is given in Theorem \ref{TCycle}. 
 
\section{The Essential Minor}
\label{sec:essential}

In this section we introduce the notion of the essential minor $(C,\delta)$ of
given graph $(G,\omega)$, which is a simple, weighted digraph that decodes the
topological structure of $G$ in a compact way and is therefore very helpful. Even though there are some crucial differences
we want to mention that similar concepts were already introduced in \cite{berge}.

\begin{definition}[parallel arcs]
Let $G$ be a graph and  $e=(u,v) \in E$ then we denote with  
\begin{eqnarray}
F^+(e)&:=&\li\{f \in E \mi \tail{f}=u\,, \head{f} = v \re\} \text{ and} \nonumber\\ 
F^-(e)&:=&\li\{f \in E \mi \tail{f}=v\,, \head{f} = u \re\} \nonumber 
\end{eqnarray}
the sets of all parallel and anti-parallel arcs and set $F(e) = F^+(e)\cup
F^-(e)$.  
\end{definition}

\begin{figure}[t!]
 \centering
 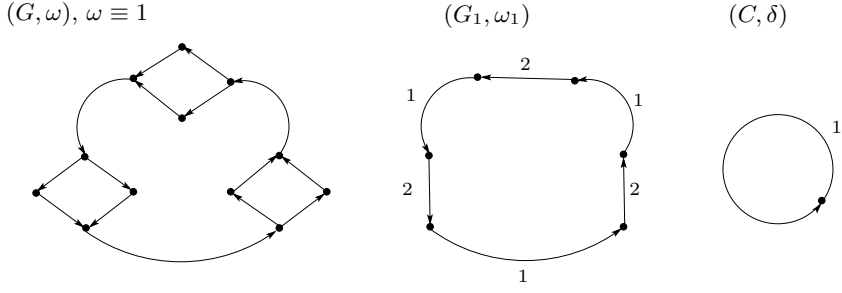
 \caption{The construction of $(C,\delta)$, with respect to $(G,\omega)$. }
 \label{F1}
\end{figure}
We recall that for a given set $A$ and an equivalence relation $\sim$ on $A \times A$ the quotient $A/_\sim$ is given by the set of all equivalence classes $[a]_\sim= \li\{x \in A \mi x \sim a \re\}$.

\begin{definition}[contracted graph] \label{def:collapse} 
Let $G =(V,E)$ be a graph and $ u,v \in V$. An equivalence relation
$\sim_{u,v}$ on $V$ is defined by $$ x \sim_{u,v} y \qquad \eq \quad x=y \quad
\text{or} \quad x,y \in \{u,v\}.$$ The equivalence class of $x\in V$ is denoted
by $[x]_{\sim{u,v}}$ and $V/_{\sim_{u,v}}$ gives the quotient of $V$ with
respect to $\sim_{u,v}$. A multiset $E/_{\sim{u,v}}$ is defined by 
\begin{align*}
E/_{\sim{u,v}} = \Big\{ (p,q) \in V/_{\sim{u,v}}\times V/_{\sim{u,v}} \mi & \exists (x,y)\in E:[x]_{\sim{u,v}}=p \,, [y]_{\sim{u,v}}=q \Big\}\,. 
\end{align*}
 The \emph{contracted graph} of $G$ with respect to $e\in E$ is
defined as the topological minor   
\begin{equation}
 G/e:= \big(V/_{\sim_{e^+,e^-}},(E\setminus F(e))/_{\sim_{e^+,e^-}} \big).
\end{equation}
If $e,f \in E$ then one can check easily that by identifying $e,f$ with their images in $G/e$, $G/f$ respectively we have $(G/e)/f= (G/f)/e$. 
Thus, the definition does not depend on the order of the contracted edges. Hence, if $G^\prime \subseteq G$ is a subgraph then 
$G/G^\prime := G/\Ec(G^\prime)$  can be defined by contracting $\Ec(G^\prime)$ in arbitrary order.    
\end{definition}

\begin{definition}[essential minor]
\label{def:ecg}
Let $G =(V,E, \omega)$ be a  positively weighted graph.  The equivalence relation
$\sim_{\Gamma}$ on $E$ is defined by $e\sim_{\Gamma}f$ if and only if $e=f$ or  there exists a
directed, branch-point-free path $(w_0,w_1,\dots,w_n,w_{n+1})$ with
$\{(w_0,w_1),(w_n,w_{n+1})\}=\{e,f\}$ and $n\in\N^+$, i.e., for $i=1,\dots,n$ it
holds that $\deg^-(w_i)=\deg^+(w_i)=1$. We represent an equivalence class
$[e]_{\sim\Gamma}$ by an arc $(u,v)$, where $u$ and $v$ coincide with the
start and endpoint of the longest directed, branch-point-free path running
through $e$.  The positively weighted graph
$$
\big(G/_\Gamma,\omega/_\Gamma)\big):= \big(V/_\Gamma, E/_\Gamma, \omega/_\Gamma)\big)
$$
is defined by 
\begin{enumerate}[i)]
	\item $V/_\Gamma:=\{v\in V \mi \deg^-(v)> 1 \vee \deg^+(v)> 1\}
		$, 
	\item $E/_\Gamma:=\li \{ (u,v)\in V/_\Gamma\times V/_\Gamma
		\mi  \exists \,\, [e]/_{\sim\Gamma} \in E_{/\sim_\Gamma} \,\,
		\text{represented by } (u,v)  \re \}$ 
	\item $\omega/_\Gamma:E/_\Gamma\lo \N^+$ with
		$\omega/_\Gamma(e):=\min_{e'\in [e]/_{\sim\Gamma}}
		\omega(e') $.
\end{enumerate}
Let $\sim_{\Phi}$ be an equivalence relation on $E$ with $e \sim_{\Phi} f \eq
e,f \in F^+(e).$  For $(G,\omega)$ the positively weighted graph 
$$(G/_\Phi,\omega/_\Phi):=\big(V/_\Phi,E/_\Phi,\omega/_\Phi\big)$$ is defined by 
\begin{enumerate}[i)]
\item $V/_\Phi :=V$.
	\item $E/_\Phi:=E/_{\sim_{\Phi}}$, where we identify each equivalence
		class $[e]_{\sim\Phi}$ with an arbitrary representative of $F^+(e)$.
	\item $\omega/_\Phi:E/_\Phi \lo \N^+ \,,\,\, \omega/_\Phi(e):=
		\sum_{e^\prime \in [e]_{\sim\Phi}} \omega(e^\prime)\,.$
\end{enumerate}
Starting with $G_0:=G$ for $k >0$ we define 
$$(G_k,\omega_k):=(G'_{k-1}/_\Phi,\omega'_{k-1}/_\Phi)\, \text{ with } (G_{k-1}',\omega_{k-1}'):=(G_{k-1}/_\Gamma,\omega_{k-1}/_\Gamma)\,.$$ 
The weighted graph $(G_K,\omega_K)$
with $G_{K}=G_{K+1}$, $K\geq 0$, is called the \emph{essential minor} of $G$
and is denoted by $(C,\delta):=(V_C,E_C,\delta)$. 
\end{definition}

\begin{example}\label{exa:essential}
In Figure~\ref{F1} the construction of the essential minor  is illustrated.  
Furthermore, for a graph $G$ with $D$ diamonds  connected in a cycle as in the
example we obtain $\deg(G) = 3$ and $\big|\Oe(G)\big|= 2^D=2^{|E|/5}=2^{|V|/4}$. 
In contrast the essential minor $C$ of such a graph satisfies $|\Oe(C)|=1$.
Hence, even though the number of cycles in $G$ increases exponentially in $|E|$ and
$|V|$ by adding further diamonds, the number of cycles in $C$ remains constant 
equal to 1 while the weight $\xi$ decodes the number of cycles of the original graph $G$.
\end{example}

Since the following results are quite canonically their simple but technical proofs are given in Appendix A. 

\begin{proposition}
\label{cor:ecg_equiv}
Let $G=(V,E,\omega)$ be a positively weighted graph with essential minor
$(C,\omega_C)$ and let $\ee \in \Pc(E)$ and $\ee_C$ be the image of $\ee$ in $(C,\delta)$. Then 
$$\ee\in \Sc(G,\omega) \eq \ee_C \in \Sc(C,\delta)\,.$$
In particular $\Omega(G,\omega)=\Omega(C,\omega_C)$.
\end{proposition}

Proposition \ref{cor:ecg_equiv} states that solving the FASP for the essential minor 
is equivalent to solving the FASP on the original graph. Even though it is possible that $(C,\delta)=(G,\omega)$ the construction might yields 
an a~priori problem size reduction in many cases as in Example \ref{exa:essential}. 

\begin{proposition}
\label{pro:construct_essential}
Let $G =(V,E,\omega)$ be a finite, connected, directed,  weighted
multigraph then we can construct $(C,\delta)$ in time $\Oc(|E|^2)$.
Furthermore, there is an algorithm with run time $\Oc(|E|^2)$ which constructs a
solution $\ee \in S(G,\omega)$ given a solution $\ee_C \in \Sc(C,\delta)$. 
\end{proposition}

\begin{remark} 
Since for given $\ee_C \in \Sc(C,\delta)$ the construction of some $\ee_G \in\FF(\ee_C)$ is easy to compute (Proposition \ref{pro:construct_essential}) Proposition \ref{cor:ecg_equiv} states that 
it suffices to solve a the weighted FASP for the essential minor instead of 
the original graph. 
As a consequence multigraphs do not need to be considered and the number of elementary cycles of the essential minor can be drastically reduced, see for instance Example \ref{exa:essential}. 
\end{remark}

\section{Subgraphs of Elementary and Simple Cycles }
\label{sec:elecyc}
There are several approaches for generating the set $\Oe(G)$ of all elementary cycles of a graph, 
see \cite{Mateti:1976} for an overview. Since there can be an exponential 
number of cycles in a graph, generating algorithms have an exponential worst
case run time.  The best algorithms available today are the ones of \cite{Tarjan} and \cite{Johnson1975} solving the problem in $\Oc\big(|\Oe(G)|(|V|+|E|)\big)$. 
Of course counting all cycles might be less expansive than generating them. However, by reducing to the Hamiltonian cycle problem, see for instance \cite{Arora},  
counting all cycles is a NP-hard problem.
For our concerns, and supposedly in many other situations, the generation of
all cycles is not necessary, but the knowledge of the arc set of all cycles including a common arc suffices.  In
the following an algorithm for determing the smallest subgraph
$G_{\el}(e)\subseteq G$ which contains all elementary cycles
that include the arc $e$ is given. 

\begin{definition}\label{theta}
Let $G=(V,E)$ be a graph and $u,v \in V$ we denote with with $P(u,v),P_{\el}(u,v)$, $P_{\si}(u,v)$ 
the set of all directed, elementary or simple paths from $u$ to $v$ respectively.
For an arc $e \in E$ we let   
\begin{align*}
 \Oe(e):= &\li\{ c\in \Oe(G) \mi c \cap e \not = \emptyset\re\} = \{e\} \cup \Pe(e^-,e^+) \,\,\,\text{and}\,\,\,\\
 \Os(e):=& \li\{ c\in \Os(G) \mi c \cap e \not = \emptyset\re\}= \{e\} \cup \Ps(e^-,e^+)
\end{align*}
be the set of all elementary and simple cycles running through $e$. If $\ee \in \Pc(E)$ then we set 
$\Oe(\ee) := \cup_{e\in\ee}\Oe(e)$, $\Os(\ee):= \cup_{e\in\ee}\Os(e)$.
Moreover, we denote with $G(u,v):=\Gc(P(u,v)), G_{\el}(u,v) := \Gc\big(P_{\el}(u,v)\big)$, $ G_{\si}(u,v):= \Gc\big(P_{\si}(u,v)\big)$ the by the corresponding paths induced graphs and with  
$G(e):= \Gc(P(u,v))$, $G_{\el}(e):=\Gc\big(P_{\el}(e^-,e^+) \cup\{e\}\big)$, 
$G_{\si}(e):= \Gc\big(P_{\si}(e^-,e^+)\cup\{e\}\big)$ by the corresponding cycles induced graphs. Moreover, $P^0(u,v),\Pe^0(u,v)$, $\Ps^0(u,v)$, $O(e)$, $O_{\el}^0(e)$, $O_{\si}^0(u,v)$, 
$G^0(e)$, $G_{\el}^0(e)$, $G_{\si}^0(e)$ shall denote the connected (and not necessarily directed) analagons of the introduced sets and graphs.  
\end{definition}

Note that $P_{\el}(u,v) \subseteq P_{\si}(u,v)$ and therefore $\Oe(e)  \subseteq \Os(e)$. 
Moreover, the graphs $G(u,v),G(e)$ can be determined in $\Oc(|E|^2)$ by applying a depht first search technique similar to 
Remark \ref{GO}. In the other cases we observe: 

\begin{theo}\label{TCycle}
 Let $G=(V,E)$ be a graph and  $u,v \in V$, $e \in E$. Then there exist algorithms which compute:
 \begin{enumerate}
 \item[i)]  The graphs $G_{\el}(u,v)$,$G_{\el}^0(u,v)$, $G_{\el}(e)$,$G_{\el}^0(e)$  in $\Oc(|E|^2)$. 
 \item[ii)] The graphs $G_{\si}(u,v)$, $G_{\si}^0(u,v)$, $G_{\si}(e)$, $G_{\si}^0(e)$  in $\Oc(|E|^4)$.
\end{enumerate}
\end{theo}
 \begin{proof} Let $p \in P_{\el}(u,v)$ then no vertex $w \in \Vc(p)$ is passed twice of $p$. Thus, for every $f \in \Ec(p)$ there is a path $q \in P(u,f^-)$ with respect to 
 $G \setminus \big(N_E^-(f^-)\setminus\{f\}\big)$. 
Vice versa if $p \in P(u,v)$ is such that $\Ec(p) \not \subseteq G_{\el}(u,v)$ then has to be a vertex $w \in \Vc(p)$, which is passed at least twice by $p$ implying that there is $f \in \Ec(p)$ such that 
\begin{equation}\label{DELf}
 P(u,f^-) = \emptyset \quad \text{with respect to} \quad  G \setminus \big(N_E^-(f^-)\setminus\{f\}\big)\,.
\end{equation}
Thus, by setting $G':= G \setminus\li\{f \in E\mi f\,\, \text{fulfills \eqref{DELf}}\re\}$, every path $p \in P(u,v)$ with $\Gc(p) \not \subseteq G_{\el}(u,v)$
is interrupted in $G'$. Hence,
$G'(u,v)= \Gc(P(u,v))$ with respect to $G'$ coincides with $G_{\el}(u,v)$. Algorithm \ref{Gel} formalizes this procedure and runs  in $\Oc(|E|^2)$ if $G$ is stored in an adjacency list, enabling us to 
test whether $P(u,v)=\emptyset$ in $\Oc(|E|)$. The other cases of $i)$ can now be solved by 
replacing $u,v$ with $e$ and directed paths or cycles with connected ones.

To show $ii)$ we add two arcs $u^*=(x,u)$, $v^*=(v,y)$, $x,y \not \in V$ denote with $G^*$ the resulting graph and consider the line graph $L(G^*)=(V_L^*,E_L^*)$.
We recall that $|V_L| =|E|$, $|E_L|\leq |E^2|$ and apply the fact that  the elementary paths of $L(G)$ are in $1$ to $1$ correspondence 
to the simple paths of $G$ and therefore $\Vc\big(G_{\el}(u^*,v^*)\big) \setminus \{u^*,v^*\} = \Ec\big(G_{\si}(u,v)\big)$. Hence $ii)$ follows analogue to $i)$.  
\qed \end{proof}

\begin{figure}[t!]
  \begin{algorithm}[H]
	\KwIn{ $G=(V,E)$, $e=(u,v)\in E$ }
	\KwOut{$G_{\el}(e)$}
	$G \leftarrow G(u,v)$\;
	$E^* \leftarrow  \emptyset$\;
        \For{$f \in E'$}{
        \If{$P(u,f^-) = \emptyset$ w.r.t. $ G' \setminus \big(N_E^-(f^-)\setminus\{f\}\big)$}{
        $E^* \leftarrow E^*\cup\{f\}$\;}}
        $G_{\el}(u,v)\leftarrow G(u,v)$ w.r.t. $G\setminus E^*$\;
	\Return $G_{\el}(u,v)$
	\caption{The induced subgraph $G_{\el}(u,v)$.\label{Gel}}
\end{algorithm}
\end{figure}

\begin{remark} Note that if $(C,\delta)$ is the essential minor of $(G,\omega)$. Then the treatment of  ``parallel'' paths is avoided by the essential minor construction. Thus, 
we expect that if $C$ is significant smaller than $G$ the run time performance will increase drastically.
\end{remark}

\section{Isolated Cycles}
\label{sec:iso}
Of course the question arises whether a solution of the FASP on $G_{\el}(e)$ can be determined independently of the 
rest of the graph. 
The notion of \emph{isolated cycles} is our starting point of investigations in this manner and as it will turn out it is    a very helpful concept of
answering this question.

We recall that a Min-$\mbox{s-t}$-Cut with source $s=u$ and sink $t=v$ is given by a set $\ee \subseteq E$ such that 
$P(u,v) = \emptyset$ in $G\setminus \ee$ and $\Omega_{G,\omega}(\ee) = \sum_{e\in\ee}\omega(e)$ is minimized.

\begin{lemma}\label{Cycle} Let $(G,\omega)$ be a weighted graph and $e \in E$. Then  there is an  algorithm, which determines 
 a solution $\ee \in \Sc\big(G_{\el}(e),\omega\big)$  of the FASP on  $(G_{\el}(e),\omega)$ in $\Oc\big(|V||E|\log(|V|)\big)$, where we slightly abused notion by still denoting $\omega$ for the the restriction of $\omega$ 
 to $G_{\el}(e)$.  
\end{lemma}
\begin{proof} Observe that by interpreting $\omega$ as a capacity function on $G_{\el}(e)$ a solution of the FASP on $\ee$ is given by $\{e\}$ or a Min-$\mbox{s-t}$-Cut $\ee$ with source $s=e^-$ and sink $t=e^+$. The option with 
the smaller weight is chosen.
Due to the famous Min-Cut-Max-Flow Theorem a Min-$\mbox{s-t}$-Cut can be determined by
solving a Max-Flow problem with respect to $\omega$ and $s=e^-$, $t=e^+$.
The algorithm of \cite{Dinic1970} solves the Max-Flow problem for arbitrary weights in time
$\Oc(|V|^2|E|)$ and can be speedend up to $\Oc\big(|V||E|\log(|V|)\big)$, \cite{Dinitz}  by using the data structure of dynamic trees. Thus the statement is proven.
\qed \end{proof}

\begin{remark} Note that if $(C,\delta)$ is the essential minor of $(G,\omega)$ then the absence of ``parallel'' paths might speeds up the time required to determine a Min-$\mbox{s-t}$-Cut drastically. 
Moreover, the Max-Flow-Problem is very well understood, yielding many 
alternatives to the algorithm of \cite{Dinic1970} and providing faster solutions in special cases, see \cite{Dinitz} and \cite{Korte} for an overview.
\end{remark}

\begin{figure}[t!]
 \centering
  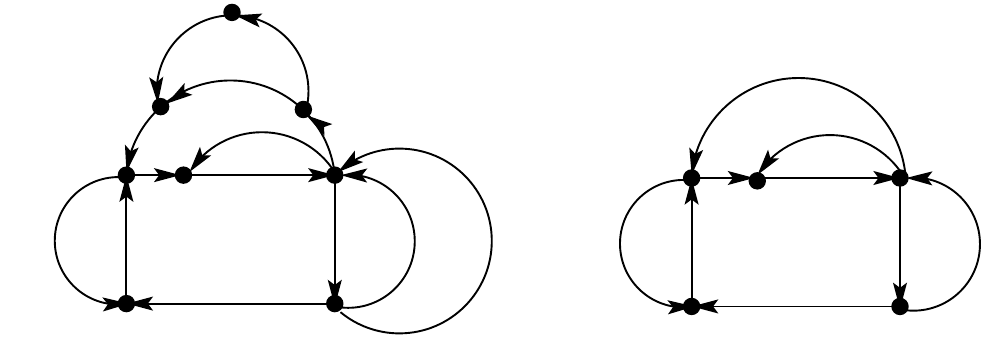
 \caption{A graph $(G,\omega)$  and its essential minor $(C,\delta)$.  The numbers at the arcs of $C$ indicate the value of 
 $\delta$. $\{e_1,f_1\}$, $\{e_2,f_3\} $, $\{e_3,f_3\}$, $\{e_4,f_4\}$, $\{e_5,e_6,e_7,f_5,f_4\}$, and $\{e_5,e_8,e_9,e_7,f_5,f_4\}$  are isolated cycles of $G$ and $I(G)= \{f_1,f_3,f_4\}$. } 
 \label{F0}
 \end{figure}

\begin{definition}
\label{def:iso}
Let $G$ be  a graph and $c \in \Oe(G)$ then we denote with
$$I(c): =\li\{ e \in \Ec(c) \mi  c\cap c^\prime  = \emptyset \,, \,\forall c^\prime
\in \Oe(G)\setminus \Oe\big(F^+(e)\big) \re\}  $$ 
the set of {\em isolating arcs} of $c$, i.e., if $e \in I(c)$ then $c$ has
empty intersection with every cycle $c^\prime$ that does not contain $e$ or an
parallel arc of $e$.
For an arc $e\in E$ or set of arcs $\ee \subseteq E$ we set 
$$O_I(e):=\li\{c \in \Oe(G) \mi e \in I(c)\re\}\,, \quad O_I(\ee)=\bigcup_{e\in\ee}O_I(e)$$
and $G_I(e):= \Gc\big(O_I(F^+(e))\big)$.  
\end{definition}

\begin{remark}\label{hierarchy}
Note, that the sets of isolated cycles possess a \emph{flat hierarchy} in the
following sense. If $e,f\in E$, $e\neq f$, with  $O_I(e)\not = O_I(f)$ then
$O_I(e)\cap O_I(f) = \emptyset $. If vice versa $O_I(e)=O_I(f)$  then by
definition we obtain $\Oe(F^+(e))=\Oe(F^+(f))$.
\end{remark}

\begin{remark}
Let $c \in O_I(G)$ be an isolated cycle and $I(c)$ the set of all isolating arcs of $c$.
If we contract $I(c)$ then the resulting graph  $G/I(c)$ fulfills 
$$\li<c',c''\re> = 0 \,, \quad \text{for all} \,\,\, c' \in O_I(e)/I(c), c'' \in  \big(O_{\el}(G/I(c)\big)\setminus \big(O_I(e)/I(c)\big)\,, $$
where $\li<\cdot,\cdot\re>$ denotes the standard scalar product on $\R^n$, $n =
\big|E/I(c)\big|$.  Thus, by detecting isolated cycles we obtain an
orthogonal splitting 
$$\Lambda\big(G/I(c)\big) = \spann\big( O_I(e)/I(c)\big)
\oplus \spann\big(O_{\el}(G/I(c)\big)\setminus \big(O_I(e)/I(c)\big)\,.  $$
Such a splitting is certainly helpful whenever one wants to find a basis of
$\Lambda(G)$, e.g.,  a minimal cycle basis of $\spann\li(O_I(e)\re)$ can be
extended to a minimal basis of $\Lambda(G)$. 
\end{remark}

Consider an isolating arc $e$ of a graph $(G,\omega)$ or its essential minor $(C,\delta)$. The isolated cycles $O_I(e)$ running through $e$ can be cut either by removing
the arc set  $\ee_0= F^+(e)$ or another feedback set $\ee_1$ of $G_I(e)$. By definition the arc set 
$\ee_0$ cuts at least the cycles in $O_I(\ee_0)$. By Remark \ref{hierarchy} 
the feedback set $\ee_1$ cuts only the isolated cycles $O_I(\ee_0)$ or is given by  $\ee_1=F^+(f)$ of another isolating arc $f \in O_I(e)$ with 
$\Oe(\ee_0) =\Oe(\ee_1)$. Thus, if the weight of $\ee_0$  equals the weight 
of a  solution of the FASP on $G_I(e)$ then there is a solution $\ee$ of the FASP on $G$  with $\ee_0 \subseteq \ee$. 
The following definition reflects this idea more formally.

\begin{figure}[t!]
  \begin{algorithm}[H]
	\KwIn{ $e=(u,v)\in E$ }
	\KwOut{$G_I(e)$}
	$G \leftarrow G_{\el}(e)$\;
	$E^*\leftarrow \li\{f \in E_{\el} \,|\, \Oe(f)\not = \emptyset \,\,\, \text{in}\,\,\ G\setminus e\re\}$\;
	$G_I(e) \leftarrow G_{\el}(e)$ with respect to $G\setminus E^*$ \;
	\Return $G_I( e )$
	\caption{The isolated cycles of a graph.\label{alg:iso}}
\end{algorithm}
\end{figure}
\begin{definition}
\label{def:resolve}
Let $(G,\omega)$ be a graph. We define a maximal list of graphs
$(G_0,\omega_0),\dots,(G_k,\omega_k)$ with $(G_0,\omega_0) =
(G,\omega)$, $(G_i,\omega_i)\neq (G_{i+1},\omega_{i+1})\,, \,\forall\, i\in[1:k-1]$ as follows.  
Let $(C_i, \delta_i)$, with $C_i=(V_{C_i},E_{C_i})$, be the essential minor of $(G_i,\omega_i)$ and 
$E^*_i\subseteq I(C_i)$ be a maximal subset of pairwise different isolating arcs of $C_i$ such that 
$\forall e \in  E^*_i:$
\begin{equation}\label{eq:resolve}
\delta(e) = \Omega\big(G_I(e),\delta_{|G_I(e)}\big)\,, \quad G_{\el}(e)\not = G_{\el}(f) \,, \,\text{whenever}  \,\,\, e\not = f\,.
\end{equation}
Then the weighted graph $(G_{i+1}, \omega_{i+1})$ is given by  
$$G_{i+1}:=(V_i,E_i)= \Gc_o\big(E_{C_i}\setminus E_i^*\big)\,, \quad \omega:= \delta_{i|E_i'}$$   
where $\delta_{i+1}:= \delta_{i|E_{i+1}}$ denotes the restriction of $\delta_i$ to $E_{i+1}$. 
If $G_{k+1}=G_k$ for some $k \in \N$ then 
$(S,\tau):=\big((V_S,E_S),\tau\big)=(C_k,\delta_k)$, $k \in \N$ is called the \emph{resolved graph} of $G$, which we shortly denote with 
$(S,\tau)=\big(S(G),\tau(\omega)\big)$.  A graph $(G,\omega)$ is called \emph{resolvable} if and only if $S = \emptyset$. 
\end{definition}
\begin{example}
 Note that $(C,\delta)$ in Figure \ref{F0} is resolvable, while $(G,\omega)$ in Figure \ref{DP} is not resolvable, but becomes resolvable for uniform weight $\omega \equiv 1$. 
\end{example}

The construction has an immediate consequence. 
\begin{theo}\label{res} Let $(G, \omega)$ be a graph. Then there is an algorithm testing whether $(G,\omega)$ is resolvable and determing a solution of the  weighted FASP on $G$ in case of resolvability in 
 $\Oc(|V||E|^3)$. 
\end{theo}
\begin{proof} Due to Proposition \ref{pro:construct_essential} the construction of the essential minor $(C_i,\delta_i)$ can be achieved in $\Oc(|E|^2)$ for every $1 \leq i \leq k$.  
Since checking whether $\Oe(f) \not = \emptyset $ can be done in $\Oc\big(|E|\big)$ by storing $G$ in an adjacency list and using depht first search to figure out whether $P(f^+,f^-)\not = \emptyset$ 
the Algorithm \ref{alg:iso} 
computes the set $G_I(e)$ in $\Oc(|E|^2)$ and therefore computing $G_I(e)$ for all arcs requires at most $\Oc(|E|^3)$ computation steps.  
Furthermore, a solution of the FASP on $G_I(e)$ can be computed due to 
Lemma \ref{Cycle} in $\Oc\big(|E|^2\big)$. Due to the fact that during the construction of $(S,\tau)$ no parallel arcs appear, we have to recompute the isolated cycles at most $|V|$ times. 
Thus,   $(S,\tau)$ can be determined in $\Oc\big(|V||E|^3\big)$. Furthermore, we can use the backtracking procedure of 
Proposition \ref{pro:construct_essential} to compute a solution of the FASP in $\Oc(|V||E|^2)$ once $(S,\tau)$ is known. \qed 
\end{proof}

\begin{figure}[t!]
 \centering
 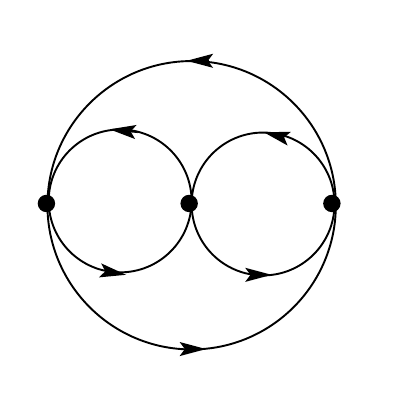
 \caption{The smallest graph without isolated cycles.}
 \label{6}
\end{figure}
Observe that Theorem \ref{res} was already stated in section \ref{main} as Theorem \ref{AA}.
However, the result leads to the question : What are fast (linear, quadratic time) checkable conditions a graph $(G,\omega)$ has to satisfy 
to be resolvable. Though, we can easily construct resolvable graphs as $(C,\delta)$ in Figure \ref{F0} or modified versions of $(C,\delta)$ by adding additional isolated cycles a characterization of 
resolvable graphs is still open.  A better understanding of the non-resolvable graphs might help to solve that problem. 
In order to investigate these graphs the class of graphs without isolated cycles at all, seems to  be interesting. Therefore, the next result might be a good starting point for further 
studies. 

\begin{proposition} The directed clique $D_3=(V,E)$ in Figure~\ref{6} is the smallest graph with $\Oe(G) \not = \emptyset$ and  $O_I(G) = \emptyset$, i.e. any non-isomorphic graph $G^\prime =(V^\prime, E^\prime)$ with 
$O(G^\prime)\not = \emptyset$ 
and $O_I(G^\prime)=\emptyset$
satisfies 
$|V^\prime |+ |E^\prime|>  |V|+|E|$\,. 
\end{proposition}
\begin{proof} If $G^\prime$ is a graph with $O(G^\prime)\not = \emptyset$, $O_I(G^\prime)= \emptyset$ then $G^\prime$ possesses at least three vertices and  $\big|O(G^\prime)\big| \geq 3$ has to hold.  We claim 
that there are at least three linear independent cycles. Indeed if $d_1,d_2,d_3 \in O(G^\prime)$ with $\lambda d_1 + \mu  d_2 +\eta d_3=0$, $ \lambda, \mu, \eta \in \Z$ then due to Remark \ref{rem:OG} we know that 
$d_1,d_2,d_3 \in \{0,1\}^{|E|}$ and no elementary cycle is subset of another. So w.l.o.g. we can assume that $ d_1 = d_2 + d_3$, which contradicts that all cycles are elementary. Thus, $\dim_\Z \Lambda(G^\prime)= |E^\prime|-|V^\prime|+1 \geq \dim_\Z O(G^\prime)  \geq 3$. This observation implies that $|E^\prime|\geq |E|$ whenever $|V^\prime| >3$. 
Since $D_3$ is an directed clique we can identify any smaller graph $G^\prime$ with a subgraph of $D_3$. It is easy to see that deletion of any arc $e \in E$ produces an isolated cycle in $D_3$. 
For instance if we delete $e_4$ then $c_3$ will be isolated, if we delete $e_5$ then $c_2$ will be isolated and so on. Hence  $G^\prime \cong D_3$. 
\qed \end{proof}

If $G$ is a non-resolvable graph then one can think about different  methods to
solve the FASP of the resolved graph $(S,\tau)$. One possibility is discussed
in the next section.

\section{The Bellman Decomposition}\label{Bell}

In this section we formulate an solution of the FASP based on a dynamic programming technique. Such an approach can be applied to optimization problems whenever there is a decomposition 
of the problem into 
subproblems which satisfy  the Bellman principle, i.e., every optimal solution consists only of optimal subsolutions. To motivate the following definitions we first consider an example. 
\begin{example}  Consider the graph $(G,\omega)$ in Figure~\ref{DP}. 
If we want to know, which arc of $c_3$ we have to cut for an optimal solution 
then this depends on the cycles $c_1,c_2$. The benefit of cutting $e_1$ instead of $e_2$ or $e_3$ is that we do not
have to cut $c_1$ anymore which costs at least $2$. Thus we introduce a new weight $\sigma$, 
which equals $\omega$ on 
$E\setminus\{e_1\}$ and is set to $\sigma(e_1)= \omega(e_1)-2 =3$ on $e_1$. Since  no other cycles than $c_3$ are cut by $d_3$ the weight of $d_3$ remains unchanged. 
Now we consider $H_{e_3,d_3} = G_{\el}(e_3)\setminus d_3$ and $H_{d_3,e_3} = G_{\el}(d_3)\setminus e_3$ and compute 
\begin{align}
\big(\omega(e_3) - \Omega(H_{e_3,d_3},\sigma)\big) -&  \big(\omega(d_3) - \Omega(H_{d_3,e_3},\sigma)\big)  =  \nonumber\\
(\omega(e_3) -3) -& (\omega(d_3)-0)  =  1 \label{sub}
\end{align}
The best solution, which contains $e_3$ is  $\{e_3,d_1\}$ 
and the best solution  containing $d_3$ is $\{d_3,e_1\}$ and we observe that  
$$\big(\omega(e_3)+\omega(d_1)\big) - \big(\omega(d_3)+\omega(e_1)\big) = 8 -7 = 1 \,.$$
Thus, the difference of the solutions coincides with the difference of the subproblems in \eqref{sub} with respect to the new weight $\sigma$.    
\end{example}

\begin{figure}[t!]
 \centering
 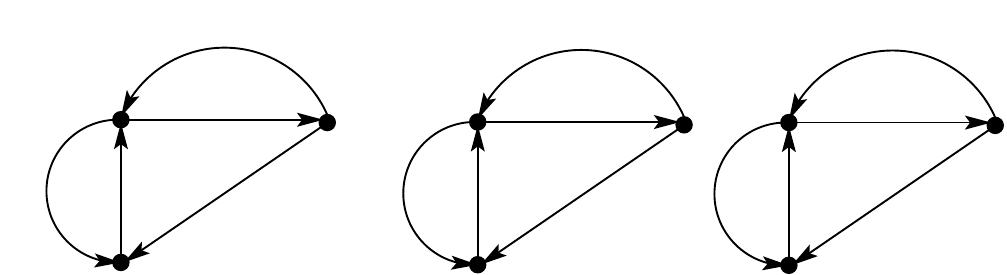
 \caption{ The computation of the relative weights with respect to $c_3$. }
 \label{DP}
\end{figure}

We  need to introduce several concepts to show that this observation remains true in general.

\begin{definition}[arc sensitivity] Let $(G,\omega)$ be a graph and $e,f \in E$, $e \not = f$ and $G_{\el}(e)=(V_e,E_e)$ be given. Then we say that $f$ is \emph{arc sensitive} to $e$ with respect to the FASP, 
denoted by $f \twoheadrightarrow e$, 
if and only if  
$$ f \in E_e \quad \text{and} \quad \Oe(f)\not = \emptyset \,\,\,\text{w.r.t.}\,\, G\setminus e \,.$$
We denote with $\mathcal{N}_{\twoheadrightarrow}(e) =\li\{f \in E \mi f\twoheadrightarrow e\re\}$   the set of all arcs, which are sensitive to $e$.
\end{definition}
Note that the arcs $f$ of an isolated cycle $c \in O_I(e)$ can not be sensitive to $e$. Thus, arc sensitivity detects arcs, which might prevent us from solving the FASP on $G_{\el}(e)$
independently from the rest of the graph. An understanding of these dependencies can be reached by understanding the \emph{meta graph} of $G$ defined in the following.

\begin{definition}[meta graph] Let $(G,\omega)$ be a graph and $c \in \Oe(G)$.  
We set $V_0=\Ec(c)$, $E_0=W_0=\emptyset$ and  
for $k\geq 1$ we define recursively $W_{k} = \cup_{i=0}^{k} V_i $ with
\begin{align*}
V_k := & \bigcup_{h\in V_{k-1}}\li\{\mathcal{N}_{\twoheadrightarrow}(h)\,\,\,\text{w.r.t.}\,\,\, \big(G\setminus(W_{k-1} \setminus \{h\}),\omega\big)\re\} \,, U_k = V_k \cup V_{k-1} \\
E_k :=&   \li\{[h,f] \in U_k\times U_k \mi  f \twoheadrightarrow h \,\,\text{w.r.t.}\,\, \big(G\setminus(W_{k-1} \setminus \{h\}),\omega\big)\re \}\,,
\end{align*} 
Stopping the recursion if $K\in \N$ is such that $V_K=\emptyset$ we introduce the simple, undirected graph $M_c := (V_{M_c},E_{M_c})=  \bigcup_{k=0}^{K}(V_k,E_k)$ as the \emph{meta graph} of $G$ 
with respect to $c$.
Furthermore, we introduce $\Cc(c)= \Gc(\Oe(V_{M_c}))\subseteq G$ as the subgraph of all \emph{arc sensitive cycles} containing $c$. 

\label{META}
\end{definition}

\begin{lemma}\label{MG} Let $G=(V,E)$ be a graph, $c\in \Oe(G)$. Then we can construct the meta graph $M_c=(V_M,E_M)$ in $\Oc(|E|^4)$.
\end{lemma}
\begin{proof} Storing $G$ in an adjacency list enables us to test whether $\Oe(e) \not = \emptyset$  in $\Oc(|E|)$ by depht first search. Due to Theorem \ref{TCycle} the graph 
$G_{\el}(e)$ can be determined in $\Oc(|E|^2)$. 
Thus, testing sensitivity requires  $\Oc(|E|^2)$.  
Due to the fact that $V_{k}\cap V_{k-1}=\emptyset $, $1 \leq k \leq K$  the  construction of $M_c$ tests for sensitivity  at most $|E|^2$ times, which yields the claimed complexity.
\qed \end{proof}

Next we define the relative weight $\sigma_{G,c,e,f}$ with respect to some $c \in \Oe(G)$ and $e,f\in \Ec(c)$.  As it will turn out $\sigma_{G,c,e,f}$ decodes which arcs of $c$ can be cutted to obtain a minimal feedback set. 
For a better clarity we firstly restrict ourselves to the case where 
the meta graph $M_c \setminus f$ is a tree. In this case, $e$ is chosen as the root and  $\sigma_{G,c,e,f}$ is given by solving the FASP for every leaf $h$ on 
$$(G_{\el}(h),\omega)\setminus \big\{ \{f\} \cup \text{inner nodes of}\,\, M_{c}\setminus f\big\}$$ 
and subtracting this value from the weight $\omega(h)$ of the predecessor of $h$. Afterwards, we delete all leafs of $M_{c}\setminus f$ and iterate this procedure till $e$ becomes a leaf. More precisely:  
\begin{figure}[t!]
\begin{minipage}[t]{1.0\textwidth}
   \begin{algorithm}[H]
	\KwIn{$(G,\omega), M, e,f \in V_M$}
	\KwOut{$\sigma_{G,M,e,f}$}
 	$M:=(V_M,E_M)\leftarrow M_{c,f}$\;
 	$G' \leftarrow G\setminus f$\; 
        $\Lc \leftarrow  \Lc_e(M):=\{h \in V_{M}\setminus\{e\} \mi \deg(h)=1\}$ \;
	$\sigma \leftarrow \omega$\; 
	\While{ $\Lc \not = \emptyset$}{
	\For{$h \in \Lc$}{$G_h \leftarrow G' \setminus \big (\{f\}\cup (V_M \setminus \Lc)\big)$ \;
	$\sigma(h) \leftarrow  \sigma(h) - \Omega(G_{\el}(h),\sigma)$ w.r.t. $G_h$\;}
	$M \leftarrow M\setminus \Lc$\; 
	$\Lc \leftarrow \Lc_e(M)$\;
 }
	\Return $\sigma$ \; 
	\vspace{2mm}
	\caption{The relative weight for meta trees. \label{sigma:tree}}
\end{algorithm}
\end{minipage}

\end{figure}

\begin{definition}[relative weight for meta trees] Let $(G,\omega)$ be a weighted graph $c\in \Oe(G)$, $e,f \in \Ec(c)$. Let $M_c$ be the meta graph of $G$ with respect to $c$ and assume that 
the connected subgraph $M_{c,f}$ of $M_c \setminus f$, which contains $e$ is a tree. Then  we define  the \emph{relative weight} of $G$ with respect to $c,e,f$
$$
\sigma_{G,M,e,f} : E \lo \R
$$
as the output of Algorithm \ref{sigma:tree} with input $\big((G,\omega),M=M_{c},e,f\big)$. 
\label{treeM}
\end{definition}
To define the relative weight in general, we have to consider all spanning trees of  $M_{c,f}$ generated  by deleting edges, which cut cycles for the first time,
seen from $e$. In Example \ref{sigex} we assert the definition for a special meta graph.  The precise definition can be found below, using the following notions: \\

 \begin{figure}[t!]
\begin{minipage}[t]{1.0\textwidth}
   \begin{algorithm}[H]
	\KwIn{$(G,\omega), M=(V_M,E_M) , e \in V_M$}
	\KwOut{$\sigma_{G,M,e}$}
	$M:=(V_M,E_M)\leftarrow M_{c,f}$\;
$W \leftarrow $ FILO $\{e\}$, $Q \leftarrow $ FILO $\{e\}$,  $K \leftarrow K(M,e)$, $U_{M_W} \leftarrow U(M,e)$\; 
\While{$W \not = \emptyset$}{
 
\eIf{$M_W$ is a tree}{
$W \mapsto h$, $Q \mapsto q$\;
$k \leftarrow d(q,p_h(M,q))$\; 
$N \leftarrow M_W\setminus D_{k-1}(M_W,q)$\; 
$\sigma(h) \leftarrow \sigma_{N,p_h(M_W,e),f}(h) $\;
$U(M_W,q) \leftarrow U(M_W,q)  \setminus\{h\}$\;
$W \leftarrow W \setminus h$\;
\eIf{$U(M_W,q) \not = \emptyset$}{
Choose $h\in U(M_W,q)$ \;
Push $h$ to $W$ \; }
{ $Q \mapsto q$, $Q \mapsto \mapsto o$\; 
$W \mapsto h$, $W \leftarrow W\setminus h$\; 
$M_W \leftarrow M_W^{\leq U(M_W,q))}$ \;
$ W\mapsto h$\; 
$\sigma(h) \leftarrow \sigma_{G,M_W,p_h(M_W,o),f}(h) $\;
$ Q \leftarrow Q \setminus q$\;}

}
{ 
$Q \mapsto q$\; 
$U \leftarrow U(M_W,q)$\;
Choose $h \in U$\; 
Push $h$ to $W$\;
Push $p_h(M_W,q)$ to $Q$\;}}
	\Return $\sigma$
	\vspace{2mm}
	\caption{The relative weight for arbitrary meta graphs. \label{sigma:gen}}
\end{algorithm}
\end{minipage}
\end{figure}

For any tree $M=(V_M,E_M)$ and any vertices $h,e \in V_M$ we denote with $p_h(M,e)$ the predecessor of $h$ with respect to root $e$. 
If  $M=(V_M,E_M)$ is an arbitrary simple, undirected graph and $q\in V_M$ then  
we consider the set  
$D_k(M,q):=\li\{w \in V_{M} \mi d(q,w)=k\re\} $ of all vertices possessing shortest path distance $k$ with respect to $q$ in $M$.
Furthermore, we consider 
\begin{align*}
 U_k(M,q) := \big\{w \in D_k(M,q) \mi &  \exists\, x \in D_k(M,q)\setminus\{w\}  \,\,\text{such that}\,\, P(w,x)\not = \emptyset \\ 
 & \text{with respect to}\,\,  M\setminus (\cup_{l=0}^{k-1} D_l) \big\},
\end{align*}
$K(M,q): = \min \{k \in \N \mi U_k(M,q)\not = \emptyset \}$ and $U(M,q):=U_{K(M,q)}(M,q)$. 
In other words: $U(M,q)$ denotes the set of vertices $w$ which cut cycles for the fist time, seen from starting point $q$.

We set  
\begin{equation}
 \widetilde M_h:= M \setminus\li\{[p_h(M,q),k] \in E_M  \mi k \in U(M,q)\setminus\{h\}\re\} 
\end{equation}
and denote with $M_h$ the connected component of $\widetilde M_h$ containing $h$. Recursively  for $n \geq 1$ and an ordered set $F=\{h_0,\dots,h_n\}$ with 
$h_n \in U(M_{h_0,\dots, h_{n-1}},q)$ we define
$M_{h_0,\dots,h_n}:= (M_{h_0,\dots,h_{n-1}})_{h{_n}}$. 
For  $h \in U(M,q)$ we consider 
$$ \widetilde M^{\leq  U(M,q)}:= M \setminus \li \{[h,k] \in E_M \mi k \in U(M,q)\,, d(q,k) \geq d(q,h) \re \} $$
and set $M^{\leq  U(M,q)}$ to be the connected component of  $ \widetilde M^{\leq  U(M,q)}$ containing $q$, which is therefore a tree.

\begin{definition}[relative weight in general] 
Let $(G,\omega)$ be a graph, $c\in \Oe(G), e,f\in \Ec(c)$ and $M_{c}$ be the meta graph of $G$ with respect to $c$. Let $\sigma_{G,M,e}$ be  the output of 
Algorithm \ref{sigma:gen}
with input $\big((G,\omega),M_c,e,f\big)$, $M=M_{c,e}$. Then, we define 
$$\sigma_{G,c,e,f}: E \lo \R \,, \quad \sigma_{G,c,e,f}(h) := \li\{\begin{array}{ll}
                                                             \sigma_{G,M,e,f}(h) & \,, \text{if}\,\, h \in \mathcal{N}_{\twoheadrightarrow}(e)\,\, \text{w.r.t.} \,\,G\setminus f\\
                                                             \omega(h) & \,, \text{else}
                                                            \end{array}\re.
$$ 
as the \emph{relative weight} of $G$ with respect to $c,e,f$.  
\label{relweight}
\end{definition}

\begin{example} Let $(G,\omega)$ be a graph, $c\in \Oe(G), e_0,e_1\in \Ec(c)$ and assume that  $M_{c,e_1}$ coincides with $M$ from Example \ref{EXW}. 
We follow Algorithm \ref{sigma:gen} 
to compute $\sigma_{G,M,e_0,e_1}: E \lo \R$. Observe that  $U(M,e_0)=\{f_0,f_1\}$ and 
$U(M_{f_1},e_0)=\{h_0,h_1\}$. The graph $M_{f_1,h_0}$ is sketched in the next picture  and turns out to be a tree. 
Now we delete all vertices which are closer to $e_0$ as 
$p_{h_1}(M,e_0)$ and obtain the graph $N$.  Next we compute the relative weight $\sigma(h_1):=\sigma_{N,p_{h_1}(M,p_{h_1}),e_1}(h_1)$ of $h_1$ with respect to $N,p_{h_1}$. 
Analogously, we compute 
$\sigma(h_0):=\sigma_{N,p_{h_0}(M,p_{h_0}),e_1}(h_0)$ and consider the graph $M^{\leq}= M_{f_1}^{\leq U(M_{f_1},e_0)}$, which is sketched in the last picture. 
Now $M^{\leq}$  is a tree and the predecessor of $f_1$ is $e_0$. Thus, we can compute 
$$\sigma_{M_{f_1}^{\leq U(M_{f_1},e_0)},e_0,e_1}(f_1) \quad \text{and  analogously} \quad \sigma_{M_{f_0}^{\leq U(M_{f_0},e_0)},e_0,e_1}(f_0)\,,$$
which finishes the computation of $\sigma_{G,M,e_0,e_1}: E \lo \R$ by replacing $\omega(f_0),\omega(f_1)$ with these weights, respectively.
 \label{sigex}
\end{example}

\begin{figure}[t!]
 \centering
  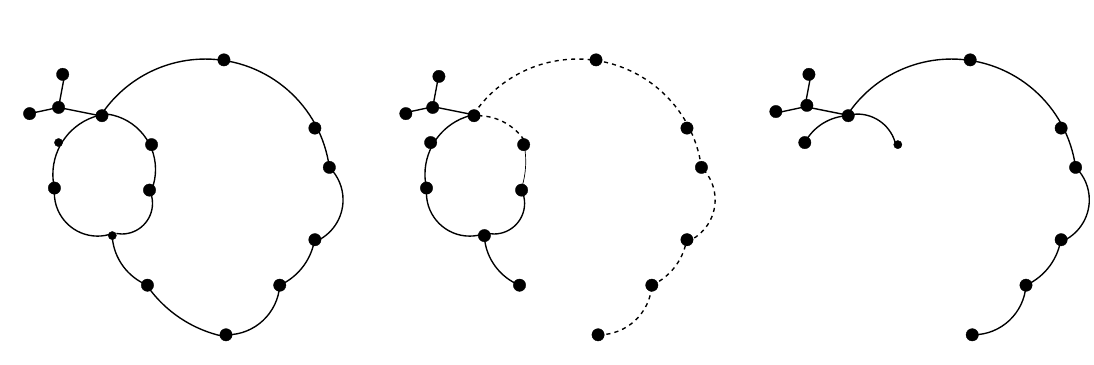
 \caption{Computation of $\sigma_{G,M,e_0}(f_1)$.  }
 \label{EXW}
\end{figure}

\begin{proposition}\label{siggi} Let $G=(V,E)$ be a graph $c \in \Oe(G)$, $e,f \in \Ec(c)$ and let the meta graph $M_{c}$ of $G$ with respect to $c$ be given. Denote with 
$$m(c,f):= \dim_{\Z_2}\big(\Lambda^0(M_{c,f})\big) = |E_{M_{c,f}}|-|V_{M_{c,f}}|+1 $$
the $\Z_2$-dimension of the cycle space $\Lambda^0(M_{c,f})$ of $M_{c,f}$.  Then 
\begin{enumerate}
 \item[i)] The computation of  $\sigma_{G,c,e,f}$ can be realized in $\Oc\big(2^{m(c,f)}|V||E|^2\log(|V|)\big)$. 
 \item[ii)] $m(c,f) \leq \dim_{\Z_2}\Lambda^0(G)= |E|-|V| +\#G$ 
\end{enumerate}
\end{proposition}
 \begin{proof} Assume that  $M_{c,f}$ is a tree. 
As already mentioned, due to \cite{Dinic1970} and \cite{Dinitz} the feedback length $\Omega(G_f,\sigma)$  can be determined in \linebreak $\Oc\big(|V||E|\log(|V|)\big)$, and has to be computed
at most $|V_{M_{c,f}}|\leq |E|$ times.
If $M_{c,f}$ is not a tree then we observe that the most expansive computation step in Algorithm \ref{sigma:gen}
is again the computation of the relative weight with respect to a certain subtree of $M_{c,f}$ (lines 8 and 19 in Algorithm \ref{sigma:gen}).  
This computation step has to be computed for every pair 
$f,h \in U(M,q)$ twice, for some $M,q$. In worst case the combination of the pairs is independent, i.e., 
every other pair still appears once $M_f$ and $M_h$ are considered. In this case the set of remaining cycles running through the remaining pairs $(h',f')$ do not contain 
the edges $h$ and $f$ and can therefore not be generated  by the cycles running through 
$(h,f)$ with respect to $\Z_2$-coefficients. Consequently, there are  at most $2^{m(c,f)}$ iterations. Together with the argumentation above this yields $i)$.

To show $ii)$ we write  $d= \{h_0,\dots,h_k\} \subseteq V_{M_{c,f}}$ as a list of connected meta vertices. Now we choose connected paths $p_i\subseteq G$, $0 =1 \dots,k+1$ connecting $h_i^-$ with $h_{i+1 \mod k+1}^+$. Then 
$c =\{\{h_0\} \cup p_0 \cup \dots, \{h_k\} \cup p_k \in \Lambda^0(G)$ is a connected cycle in $G$. If $d_1,\dots,d_m \in \Lambda^0(M_{c,f})$ is a set of $\Z_2$-linear independent meta cycles then regardless 
of choices for the paths $p_i$ representing an meta edge the corresponding cycles 
$c_1,\dots,c_m\in \Lambda^0(G)$ are $\Z_2$-linear independent in $\Lambda^0(G)$. Thus, $m(c,f)$ is bounded by the $\Z_2$-dimension of $\Lambda^0(G)$,  proving $ii)$. 
\qed \end{proof}

\begin{remark} Note that a graph $G=(V,E)$ with $c \in \Oe(G)$ such that $M_c$ coincides with $M$ in Figure \ref{relweight} can be easily constructed by choosing a starting cycle $c$ and additional cycles
$c_1,c_2$ intersecting with $c$ in $f_0,f_1$, respectively. Then we continue this process by follwing $M$ for the ramining cycles. Hence, the set of graphs $G$ with cycle $c$ and meta graphs $M_c$ such that 
$\dim_{\Z_2\Lambda^0}(M)$ is small, is actually a huge set. 
\label{Msmall}
\end{remark}

Indeed the relative weights satisfy a Bellman condition, which can be formulated as follows. We recall that $\Gc_o(G)\subseteq G$ denotes the subgraph induced by  all cycles of $G$ and state:   
\begin{theo}\label{Bellmann} Let $(G,\omega)$ be a weighted graph  $c \in \Oe(G)$, $e,f\in \Ec(c)$  and 
$H_{e,f} = \Gc_o\big(G_{\el}(e)\setminus f\big),H_{f,e} = \Gc_o\big(G_{\el}(f)\setminus e\big)$ then 
\begin{align}\label{bel}
  \big(\omega(e) - \Omega(H_{e,f},\sigma_{e})\big) - &\big(\omega(f) - \Omega(H_{f,e},\sigma_{f})\big)  =  \nonumber\\
  \big(\omega(e)+ \Omega(G\setminus e ,\omega)\big) - &
  \big(\omega(f) +\Omega(G \setminus f,\omega)\big)\,,
\end{align}
where we shorten $\sigma_e=\sigma_{G,c,e,f},\sigma_f=\sigma_{G,c,f,e}$ and slightly abuse notation by still denoting  $\sigma_{e},\sigma_{f},\omega$ 
for the restriction of the arc  weights to the corresponding subgraphs, respectively. 
\end{theo}
\begin{remark} Note that \eqref{bel} is a quite comfortable way of formulating the Bellman principle, i.e., 
though we do not know the values of $\Omega(G\setminus e,\omega)$ and
$\Omega(G\setminus f,\omega)$ we know 
that if  $e$ maximizes 
$$\omega(e) - \Omega(H_{e,h},\sigma_{G,c,e,h})  - \li ( \omega(h) - \Omega(H_{h,e},\sigma_{G,c,h,e})\re) \,,    $$
for all on $\Ec(c)$.  Then 
$\omega(e) + \Omega(G\setminus e,\omega)  = \Omega(G,\omega) $.
Thus,  $\{e\}$ can be extended to a global optimal solution. 
Maybe this relative formulation 
can be applied also to other problems for which one wants to use  a dynamic programming technique. The described observation is also used in the proof
of Theorem \ref{Bellmann}. 
\end{remark}

In addition to the observation above the following statement is needed  to prove Theorem \ref{Bellmann}. 
\begin{lemma}  Let $G=(V,E)$ be a graph, $c \in \Oe(G)$ and $e,f \in \Ec(c)$ such that there is  $c' \in \Oe(e)\setminus\Oe(f)$ and 
$p \in \Ec(c')$. Then
\begin{enumerate}
 \item[i)] $\sigma_{G,c,e,f}(h) = \sigma_{G\setminus f,c',e,p}(h) \quad \text{for all} \,\,\,h \in \Ec\big(H'_{e,p}\big)$.
 \item[ii)] $\sigma_{G,c,e,f}(p) = \omega(p) - \Omega\big(H'_{p,e},\sigma_{G\setminus f,c',p,e}\big) $,
 where $H'_{p,e}$ is understood with respect to $G'$.
\end{enumerate}
\label{CE}
\end{lemma}
\begin{proof} To verify  $i)$ and $ii)$ one has to follow directly Definitions \ref{META}, \ref{treeM} and \ref{relweight}, which  is left to the reader. 
\qed \end{proof}

\begin{proof}[of Theorem \ref{Bellmann}] If $\Oe(e)= \Oe(f)$ with respect to $G$ then $\sigma_{G,c,e,f}=\sigma_{G,c,f,e}$ and $\Omega(G\setminus e ,\omega)\big)= \Omega(G\setminus f ,\omega)\big)$ and therefore 
the claim follows. 
Now we argue by induction on $|\Oe(G)|$. If $|\Oe(G)| =1$ then there is only one totally isolated cycle and therefore  $\Oe(e) =\Oe(f) = \{c\}$. 
Thus, we are in a special 
case of the situation above and obtain the claim. 
Now assume that   $|\Oe(G)| >1$ and $\Oe(f)\subsetneq \Oe(e)$. We consider $G':=G\setminus f$ and observe that $|\Oe(G')|<|\Oe(G)|$. 
We choose $c' \in \Oe(G')$ with $e \in \Ec(c')$ and choose $p \in \Ec(c')$ such that 
\begin{equation}\label{max}
  \big(\omega(e) -\Omega(H'_{e,p},\sigma_{G',c',e,p})\big)  - \big( \omega(p) -\Omega(H'_{p,e},\sigma_{G',c',p,e}) \big) 
\end{equation} 
is maximized on $c'$, where $H'_{e,p}$ ,$H'_{p,e}$ are understood with respect to $G'$. Thus, following Remark \ref{bel} there holds  
\begin{equation}\label{opt}
 \omega(p) + \Omega(G'\setminus p, \omega) = \Omega(G',\omega)\,. 
\end{equation}
We set $\sigma'_e :=\sigma_{G',c',e,p}, \sigma'_p :=\sigma_{G',c',p,e}$ then by induction and \eqref{opt} we compute 
\begin{align}
\big(\omega(e) -\Omega(H'_{e,p},\sigma'_{e})\big)- \big(\omega(p) -\Omega(H'_{p,e},\sigma'_{p})\big) = & \big(\omega(e)+ \Omega(G'\setminus e ,\omega)\big)   \nonumber  \\
 - &  \big(\omega(p) +\Omega(G' \setminus p,\omega)\big) \nonumber   \\ 
                                                           = &  \, \omega(e) + \Omega(G\setminus \{ e,f\},\omega) \nonumber \\ 
                                                           -&  \Omega(G\setminus f,\omega) \,.  \label{I} 
\end{align}
On the other side we consider
$G''=(V'',E''):= H_{e,f}$ with  the arc weight  
$$ \gamma : E'' \lo \R^+ \,, \quad  \gamma(h):= \sigma_{G,c,e,f}(h) \,. $$ 
Now observe that $H''_{p,e} = \emptyset $ with respect to $G''$  and by Lemma \ref{CE} $i)$ we have  $\sigma''_{e}:=\sigma_{G'',c',e,p}(h) = \sigma_{G',c,e,p}(h) = \gamma(h)$ for all 
$h \in \Ec (H''_{e,p})$, where $H''_{e,p}$ 
is understood with respect to $G''$. 
Moreover, $\gamma(e)=\omega(e)$ and therefore  
\begin{eqnarray}
\big(\gamma(e) -\Omega(H''_{e,p},\sigma''_{e})\big)-  \big(\gamma(p) -\Omega(H''_{p,e},\sigma''_{p})\big)
&=&  \gamma(e) - \Omega(H''_{e,p},\gamma) - \gamma(p)  \nonumber \\
&=&  \gamma(e) - \Omega(G''\setminus p,\gamma)  - \gamma(p) \nonumber  \\
 &=& \omega(e)  - \Omega(H'_{e,p},\sigma_{G',c,e}) \nonumber  \\ 
 &-&\sigma_{G,c,e}(p)   \label{II}
\end{eqnarray}
Due to Lemma \ref{CE} $ii)$ we have that $ \sigma_{G,c,e,f}(p) = \omega(p) -\Omega(H'_{p,e},\sigma'_{p})$, $\sigma_p'=\sigma_{G',c',p,e}$. 
Inserting this fact in \eqref{II} gives 
\begin{align}
\big(\gamma(e) -\Omega(H''_{e,p},\sigma''_{e})\big)-  \big(\gamma(p) -\Omega(H''_{p,e},\sigma''_{p})\big) = &  \big(\omega(e) - \Omega(H'_{e,p},\sigma'_{e})\big) \nonumber  \\ 
  - & (\omega(p) - \Omega(H'_{e,p},\sigma'_{p})\big) 
\label{III}
\end{align}
On the other, by \eqref{opt} we have that \eqref{III}  
is maximized on $c'$. Thus, again by induction 
\begin{eqnarray}
 \big(\gamma(e) -\Omega(H''_{e,p},\sigma''_{e})\big) - (\gamma(p) - \Omega(H''_{p,e},\sigma''_{p})\big)  &=&
 \big(\gamma(e)+ \Omega(G''\setminus e,\gamma)\big) \nonumber\\ 
 &-& \big(\gamma(p) + \Omega(G''\setminus p ,\gamma) \big) \nonumber \\
 &=& \omega(e) -  \Omega(G'',\gamma)  \label{IV}  \\
 &=& \omega(e) -  \Omega(H_{e,f},\sigma_{G,c,e,f}) \nonumber 
 \end{eqnarray} 
Thus, by combining \eqref{III} with \eqref{I} and again \eqref{III} with \eqref{IV} we obtain
\begin{equation} \label{V}
   \omega(e) -  \Omega(H_{e,f},\sigma_{G',c,e}) =  \omega(e) + \Omega(G\setminus \{ e,f\},\omega) - \Omega(G\setminus f,\omega) 
\end{equation}
If $\Oe(f) \subseteq \Oe(e)$ then $ \Omega(H_{f,e},\sigma_{G,c,f,e}) = 0$ and  $\Omega(G\setminus \{ e,f\},\omega) =  \Omega(G\setminus e,\omega)$. Thus, by \eqref{V}
this yields the claim. If $\Oe(f) \subsetneq \Oe(e)$ then the analogous of \eqref{V} with respect to $f$ yields 
\begin{eqnarray*}
   \big(\omega(e) - \Omega(H_{e,f},\sigma_{e})\big) - \big(\omega(f) - \Omega(H_{f,e},\sigma_{f})\big) &=& 
   \omega(e) + \Omega(G\setminus e,\omega)  \\
  + \Omega(G\setminus \{ e,f\},\omega)   - \Omega(G\setminus \{ f,e\},\omega)       &-& \omega(f) - \Omega(G\setminus f,\omega)  \,.
\end{eqnarray*}
Since $\Omega(G\setminus \{ e,f\},\omega)   - \Omega(G\setminus \{ f,e\},\omega) =0$ this finishes the proof. 
\qed \end{proof}

We consider the Algorithms  \ref{CUT},\ref{CR}, denote with $\mathrm{output}(A)$ the set of all possible outputs an algorithm $A$ can produce and conclude :    
\begin{corollary} Let  $(G,\omega)$ be a  graph then the algorithm \emph{CUT} is exact and complete with respect to the FASP, i.e,  
$$\mathrm{output}(\text{\emph{CUT}})= \Sc(G,\omega)\,,$$
while the algorithm \emph{CUT \& RESOLVE} is exact, i.e., 
$$\mathrm{output}(\text{\emph{CUT \& RESOLVE}}) \subseteq  \Sc(G,\omega)\,.$$ 
Moreover, there is $m \in  \N$ such that \emph{CUT } and \emph{CUT \& RESOLVE} possess run times 
$\Oc\big(2^{m}|E|^4\log(|V|)\big)$, where the parameter $m \leq |E|-|V|+1$ can be determined in $\Oc\big(|E|^3)\big)$.  
\end{corollary}
Indeed the corollary proves Theorem \ref{B} and by Remark \ref{Msmall} the set of graphs with small $m \in\ N$, $m << |V|$ is a huge set. 
\begin{figure}[t!]
\begin{minipage}[t]{1.0\textwidth}
	\begin{algorithm}[H]
	\KwIn{$G=(V,E,\omega)$}
	\KwOut{$\ee \in \Sc(G,\omega)$}
	$\ee =\emptyset$\;
	\While{ $\exists f \in E$ with $P(f^-,f^+) \not = \emptyset$}{
	$F \leftarrow \Ec(G_{\el}(f))$\;
	\While{ $\exists e \in F$ with $P(e^-,e^+) \not = \emptyset$}{
	Choose $p \in \Pe(e^-,e^+)$\;
	$c \leftarrow \{e\} \cup p$ \;
	$k \leftarrow  \argmax^*_{h \in \Ec(c)}\big(\big(\omega(e)- \Omega(H_{e,h},\sigma_{G,c,e,h})\big)-\big(\omega(h)- \Omega(H_{h,e},\sigma_{G,c,h,e})\big)\big)$\; 
	$\ee \leftarrow \ee \cup \{k\}$\;
	$(G,\omega) \leftarrow (G\setminus k, \omega)$\;
	$F\leftarrow F \setminus \{k\}$\;}}
	\Return $\ee$ 
	\caption{ \emph{CUT} \label{CUT}}
	\end{algorithm}
\end{minipage}
\end{figure}

\begin{proof}  If  $\ee \in \Sc(G,\omega)$  and $e\in \ee$ then $\ee \setminus \{e\}$ solves the minimal FASP
on $G(E\setminus\{e\})$.  
Thus, the exactness and completeness statements follow directly from Theorems \ref{Bellmann},\ref{res}.
For $c \in \Oe(G)$ we set $m(c):= \dim_{\Z_2}\Lambda^0(M_c)$. Then $m(c) \geq m(c,f)$ for all $f \in \Ec(c)$  with $m(c,e)$ and  Proposition \ref{siggi} implies that $m(c) \leq |E|-|V|+1$ holds.
Let $c'\in \Cc(c)$ with 
$c'\cap c \not = 0$, be any cycle in the 
component of all arc connected cycles containing $c$, introduced in Definition \ref{META}. Then  $m(c') \leq m(c)$ on $G\setminus e$ for every $e \in \Ec(c)$. Thus,
as long as at least one arc $e\in \Ec(c)$ was deleted  the maximal number of appearing $\Z_2$-linear independent meta cycles appearing for 
the computation of $\sigma_{G\setminus e,c',h,k}$, $h,k \in \Ec(c')$ are bounded by $m(c)$. Thus, by setting $G_0=G$, choosing a cycle $c_0 \in \Oe(G_0)$, determing $M_{c_0} = (V_{M_{c_0}},E_{M_{c_0}})$ and considering 
$c_k \in \Oe(G_k)$, $G_k = G_{k-1}\setminus V_{M_{c_{k-1}}}$, $k \geq 1$ we obtain cycles $c_0,\dots,c_n$, $n \leq |E|$ with $\Cc(c_i)\cap \Cc(c_j)= \emptyset$ and 
$\cup_{i=0}^n\Cc(c_i) = \Oe(G)$. 
Since $\Ec\big(\Cc(c_i)\big)\cap \Ec\big(\Cc(c_j))=\emptyset$ 
the parameter $m:=\max_{i=0,\dots,n}m(c_i)$ can be determined in $\Oc\big(|E|^3\big)$ due to Lemma \ref{MG}. 

Since the algorithm \emph{CUT} computes $\sigma_{G,c,e,h}, \sigma_{G,c,h,e}$ for fixed $e \in \Ec(c)$ and all $h \in \Ec(c)$ and $|\Ec(c)|\leq |V|$, cuts the right arc  and repeats the computation at most 
$|E|$ times by observing that $\Oc(|V|^2)=\Oc(|E|)$ the run time of the algorithm \emph{CUT} can be estimated as claimed.  
Recall, that due to Theorem \ref{res} the resolved graph can be computed in $\Oc\big(|V||E|^3\big)$. Therefore, the analogous argumentation 
yields the claimed run time for the algorithm  \emph{CUT \& RESOLVE}. 
\qed \end{proof}

Due to the fact that the FASP is NP complete, as expected our approach depends exponentially on some parameter, which in our approach is the number $m$ of liner independent meta cycles. 
In cases where $m$ is large we have to use another method to solve the FASP or use a heuristic.

%

\begin{figure}[t!]
\begin{minipage}[t]{1.0\textwidth}
	\begin{algorithm}[H]
	\KwIn{$G=(V,E,\omega)$}
	\KwOut{$\ee \in \Sc(G,\omega)$}
	$\ee =\emptyset$\;
	\While{ $\exists f \in E$ with $P(f^-,f^+) \not = \emptyset$}{
	$F \leftarrow \Ec(G_{\el}(f))$\;
	\While{ $\exists e \in F$ with $P(e^-,e^+) \not = \emptyset$}{
	Choose $p \in \Pe(e^-,e^+)$\;
	$c \leftarrow \{e\} \cup p$ \;
        $ k \leftarrow  \argmax^*_{h \in \Ec(c)}(\omega(e)- \Omega(H_{e,h},\sigma_{G,c,e,h})-\omega(h)- \Omega(H_{h,e},\sigma_{G,c,h,e}))$\; 
        $\ee \leftarrow \ee \cup \{k\}$\;
        $(G,\omega) \leftarrow \big(S(G\setminus k),\tau(\omega)\big)$\;
	$F\leftarrow \Ec\big((G_{\el}(e)\big)$ w.r.t. $G$\;}}
	\Return $\ee$ 
	\caption{\emph{CUT \& RESOLVE}\label{CR}}
	\end{algorithm}
\end{minipage}
\end{figure}

\section{Valid Greedy Solutions }\label{Greedy}

As for instance  shown in \cite{greedy}  a greedy solution for the FASP needs not to be optimal. We give a criterium on solutions which guarantees optimility. Moreover, 
we can estimate the failure of every sub optimal solution. Finally, we suggest a heuristic given by a hybrid technique of the already presented approaches. 

For given graph $(G,\omega)$ we consider the functions 
$$\theta_G, \varphi_G  : E \lo \N\,, \quad \theta(e) := |\Oe(e)| \,, \quad \varphi_G(e):= \big|\Ec\big(G_{\el}(e)\big)\big| \,.$$ 
Recall, that due to \cite{Arora} determing $\theta_G(e)$ is a NP-hard problem and the results of \cite{Tarjan} and \cite{Johnson1975} solving the problem in $\Oc\big(|\theta_G(e)|(|V|+|E|)\big)$, 
where $\theta_G(e)$ can depend exponentially on $G$.
However, \cite{Roberts} could establish efficent and close estimations of the number of $s-t$ paths. 
Since $|\Oe(e)|= |P_{\el}(e^-,e^+)|$ the result enables us to determine $\theta_G(e)$ efficently, with small failure. In contrast, $\varphi_G$ can be determined in $\Oc(|E|^2)$ due to Theorem \ref{TCycle}.

\begin{definition}
Let $(G, \omega)$ be a given weighted graph with $G=G_o$.  We introduce the efficent weights 
\begin{equation*}
\xi_{G,\omega}, \eta_{G,\omega}: E \lo \Q^+ \,, \quad \xi_{G,\omega}(e):= \frac{\theta_{G}(e)}{\omega(e)}\,, \,\,\,  \eta_{G,\omega}(e):= \frac{\varphi_{G}(e)}{\omega(e)}
\end{equation*}
and set $\omega_{\max}(G,\omega):= \max_{e \in E}\omega(e)$, $\theta_{\max}(G)= \max_{e\in E}\theta_G(e)$, $\varphi_{\max}(G):=\max_{e\in E}\varphi_G(e)$,  $\xi_{\max}(G,\omega):= \max_{e \in E}\xi(e)$, $\eta_{\max}(G,\omega):= \max_{e \in E}\eta(e)$, 
and $\mu(G,\omega) =  \li\lceil\frac{|\Oe(G)|}{\xi_{\max}(G,\omega)}\re\rceil$, $\upsilon(G,\omega)= \li\lceil\frac{|E|}{\eta_{\max}(G,\omega)}\re\rceil$, where $\li\lceil\cdot\re\rceil$ denotes the Gauss-bracket. 
\end{definition}
\begin{theo} \label{E} Let $(G,\omega)$ be a given graph. Then  
\begin{equation}\label{MU}
 \Omega(G, \omega) \geq \max\big\{\mu(G,\omega), \upsilon(G,\omega)\big\}\,.
\end{equation}
Moreover, there are infinitely many weighted graphs $(G,\omega)$ with $\mu(G,\omega)=\Omega(G,\omega)$ or $\upsilon(G,\omega)=\Omega(G,\omega)$.
\end{theo}
\begin{proof}
Let  $\ee=\{e_1,\dots,e_n\} \in \Sc(G,\omega)$ be an arbitrarily ordered  solution of the weighted FASP. We set $ G_0 = G$ and $G_i=(V_i,E_i):=G\big(E\setminus\{e_1,\dots,e_{i-1}\}\big)$, for $i\geq 1$ 
and  denote with $\omega_{G_i}$ the corresponding restriction of 
$\omega $ to $E_i$. Now due to the fact that $\xi_{G_i,\omega_i}(e)\leq \xi_{G,\omega}(e)$ for all $e \in E_i, 1\leq i\leq n$ we obtain 
\begin{equation*}
 \Omega(G, \omega) = \sum_{i=1}^n \omega_{G_i}(e_i) = \sum_{i=1}^n \frac{\theta_{G_i}(e_i)}{\xi_{G_i}(e_i)} \geq \frac{1}{\xi_{\max}(G,\omega)}\sum_{i=1}^n\theta_{G_i}(e_i)=  \frac{|\Oe(G)|}{\xi_{\max}(G,\omega)}\,. 
\end{equation*}
Since $\Omega(G, \omega) \in \N$ this proves \eqref{MU}.
Now let $(G,\omega)$ be a graph with $\omega=1$ and $\Oe(G) = \{c_0, \dots, c_n\}$, which are arranged path like, i.e., 
$|\Ec(c_i\cap c_{i+1})| =1$ and $|\Ec(c_i\cap c_{j})| =0$ if $j \not = i$. Then one verifies easily that $\mu(G,\omega)=\Omega(G,\omega)$. By replacing $\theta_G$ with $\varphi_G$ and 
$\mu(G,\omega)$ with $\upsilon(G,\omega)$ the exact same argumentaion yields the remaining claim.  
\qed \end{proof}
\begin{remark}
Note, that of course there are many more graphs with $\mu(G,\omega)=\Omega(G,\omega)$ or $\upsilon(G,\omega)=\Omega(G,\omega)$ then those used in the proof above. 
Nevertheless, it is hard to give a good condition on a graph such that 
$\mu(G,\omega)=\Omega(G,\omega)$ or $\upsilon(G,\omega)=\Omega(G,\omega)$ holds. For instance in \cite{greedy} an example of a planar graph is given, where this is not the case. 
Certainly, the  lower bounds can be used to  improve the performance of a 
variety of algorithms solving the FASP or to control the quality of a heuristic as the Algorithms \ref{GR},\ref{GRR}.
\end{remark}
We consider the heuristics  \emph{GREEDY-CUT} and  \emph{GREEDY-CUT \& RESOLVE} and discuss their properties. 
\begin{figure}[t!]
\begin{minipage}[t]{.48\textwidth}
\begin{algorithm}[H]
	\KwIn{$G=(V,E,\omega)$}
	\KwOut{Feedback set $\ee \subseteq E$ }
	$\ee =\emptyset$\;
	\While{$\xi_{\max}(G) \not = 0$ }{
			$\ee \leftarrow \ee \cup \argmax^*_{f\in E }\xi(f)$\;
			$(G,\omega) \leftarrow (G \setminus \ee,\omega)$\;
        }
	\Return $\ee$
	\caption{ \emph{GREEDY-CUT}\label{GR}}
	\vspace{4.25mm}
\end{algorithm}
\end{minipage}
\begin{minipage}[t]{.45\textwidth}
\begin{algorithm}[H]
	\KwIn{$G=(V,E,\omega)$}
	\KwOut{Feedback set $\ee \subseteq E$ }
	$\ee =\emptyset$\;
	\While{$\xi_{\max}(G) \not = 0$ }{
			$\ee \leftarrow \ee \cup \argmax^*_{f\in E }\xi(f)$\;
			$(G,\omega) \leftarrow \big(S(G \setminus \ee),\tau(\omega)\big)$\;
        }
	\Return $\ee$
	\caption{ \emph{GREEDY-CUT \& RESOLVE}\label{GRR}}
\end{algorithm}
\end{minipage}
\end{figure}
\begin{proposition}\label{APX} Let $(G,\omega)$ and $\ee$ be a solution of \emph{GREEDY-CUT} or  \emph{GREEDY-CUT \& RESOLVE} with respect to the effective weigth $\xi$. Then 
\begin{equation}\label{app}
\frac{\Omega(G,\omega)}{\Omega_{G,\omega}(\ee)}\geq \frac{\omega_{min}(G)}{\omega_{\max}(G)\cdot \theta_{\max}(G)}  \geq \frac{\omega_{min}(G)}{\omega_{\max}(G)\cdot |\Oe(G)|} \,. 
\end{equation}
If in particular $\omega \equiv1$ then $|\ee|\leq |E|/2$. 
\end{proposition}
\begin{proof} Certainly, it suffices to prove the first estimate in \eqref{app}. We show the claim for a solution $\ee=\{e_1,\dots,e_n\}$ of \emph{GREEDY-CUT}. Assume that $\ee$ is ordered with respect to appearing arcs, 
set  $ G_1 = G$ and $G_i=(V_i,E_i):=G\big(E\setminus\{e_1,\dots,e_{i-1}\}\big)$, for $i\geq 2$ 
and  denote with $\omega_{G_i}$, $\theta_{G_i}$ and $\xi_{G_i}$ the corresponding restrictions of $\omega $, $\theta_G$, $\xi_G$ to $E_i$. Then we compute 
$$ \Omega_{G,\omega}(\ee) = \sum_{i=1}^n \omega_i(e_i) = \sum_{i=1}^n \frac{\theta_{G_i}(e_i)}{\xi_{G_i}(e_i)} \leq \frac{1}{\xi_{G_n}(e_n)}|\Oe(G)| \,.$$
Since $\xi_n(e_n) = \xi_{\min}(\ee)$ we use Theorem \ref{E} to compute  
 $$\frac{\Omega(G,\omega)}{\Omega_{G,\omega}(\ee)} \geq \frac{\xi_{G_n}(e_n)|\Oe(G)|}{\xi_{\max}(G)|\Oe(G)|} \geq \frac{\omega_{min}(G)\cdot \theta_{G_n}(e_n)}{\omega_{\max}(G)\cdot \theta_{\max}(G)} \geq 
 \frac{\omega_{min}(G)}{\omega_{\max}(G)\cdot \theta_{\max}(G)} $$ and the claim follows. A proof of the statement for \emph{GREEDY-CUT \& RESOLVE} can be given by an easy adaption of the argument above 
 and is left to the reader. 
 Now let $\omega \equiv 1 $ then we argue by 
 induction on $|\ee|$ to show that $|\ee|\leq |E|/2$ for both algorithms. If $|\ee|=1$ then due to the fact that 
$G$ possesses no loops the claim follows. 
Now let $|\ee|>1$ we order $\ee=\{e_1,\dots,e_n\}$ with respect to appearance and consider $\ee_1:=\{e_1\}$, $\ee_{2}=\ee\setminus\{e_1\}$ and $G_1:=(V_1,E_1)= G(\Oe(e_1))$, $G_2:=(V_2,E_2)=G(\Oe(\ee_2))$. 
By induction we have $|\ee_2|\leq |E_2|/2$. Consider $G / (E_1 \cap E_2)$, delete all appearing loops and denote the resulting  graph with $G_1^*$. If $\Oe(G_1^*)=\{c\}$ then all cycles of
$\Oe(G)$ are totally isolated and the claim follows by triviality. If $|\Oe(G_1^*)|>1$ then $|\ee_1|\leq |E_1^*|/2$.  Since $E_1^* \cap E_2 = \emptyset$ this implies that 
\begin{equation*}
|\ee|= |\ee_{1}| +|\ee_2|\leq |E_{1}^*|/2 + |E_{2}|/2 \leq |E|/2
\end{equation*}
as claimed.
\qed \end{proof}

Be replacing $\theta_G$, with $\varphi_G$ and $\xi_{G,\omega}$ with $\eta_{G,\omega}$ in \emph{GREEDY-CUT} or  \emph{GREEDY-CUT \& RESOLVE} the analouge argumentaion yields. 
\begin{proposition} Let $(G,\omega)$with $G=G_o$ and $\ee$ be a solution of \emph{GREEDY-CUT} or  \emph{GREEDY-CUT \& RESOLVE} with respect to the effective weigth $\eta$. Then 
\begin{equation*}
\frac{\Omega(G,\omega)}{\Omega_{G,\omega}(\ee)}\geq \frac{\omega_{min}(G)}{\omega_{\max}(G)\cdot \varphi_{\max}(G)}  \geq \frac{\omega_{min}(G)}{\omega_{\max}(G)\cdot |E|} \,. 
\end{equation*}
If in particular $\omega \equiv1$ then $|\ee|\leq |E|/2$. 
\end{proposition}

\begin{example} Consider the directed clique $D_3$ from Figure~\ref{6} with constant weight $w\equiv 1$. Then $D_3$ coincides with its resolved graph and regardless of possible choices 
every candidate $\ee$ the algorithm  \emph{GREEDY-CUT} or  \emph{GREEDY-CUT \& RESOLVE}  proposes, satisfies $|\ee|=3$. Since $\mu(D_3) =3$ every candidate is optimal. 
\end{example}

Summarizing our results so far  the heuristics \emph{GREEDY-CUT} or \emph{GREEDY-CUT \& RESOLVE} solve the FASP with controlled variance in 
$\Oc(|E|^4)$, due to Theorem \ref{TCycle}, in case of effective weight $\eta$  and in $\Oc(f_\theta|E|^2)$ in case of effective weight $\xi$, where $f_\theta$ shall control the computation steps of $\theta_G(e)$, $\forall e \in E$. 
Even if we approximate $\theta(e)$ by the method of \cite{Roberts} the resulting algorithm remains an efficient heuristic.   
However, possibly there is a more accurate method available, given by a hybrid algorithm of the methods introduced in this article. 
We expect that an implementation of this strategy yields a fast and precise general FASP-SOLVER, which due to section \ref{Prel} is therfore also a FVSP-SOLVER.

\begin{strat}\label{S} For given weighted  graph $(G,\omega)$ 

\begin{enumerate}
 \item[1.] Compute the resolved graph $(S,\tau)$.  
 \item[2.] Choose a cycle $c \in \Oe(S)$ and compute the meta graph $M_{c}$. 
 \item[3a.] If the number of meta cycles $m(c)$ is large determine a ``good'' feedback vertex set $\nu_{M_{c}}$ of $M_{c}$, with respect to the vertex weight 
        $$\omega_{M_c}: V_{M_c}\lo \R^+\,, \quad 
       \omega_{V_c}(v) = \frac{1}{\omega(v)} \,, $$ 
       using one of the known or  presented methods. 
 \item[3b.] Alternatively, compute a maximal spanning tree $T_{M_c}=(V_{m_c},E_{T_c})$ with respect to the arc weight weight 
       $$\omega_{M_c}: E_{M_c}\lo \R^+\,, \quad 
       \omega_{M_c}(h) = \frac{1}{\omega(e)+ \omega(f)} \,, $$ 
       where $h=[e,f]$, $e,f, \in E$ and set $\ee_{M_c}=E_{M_c} \setminus E_{T_c}$. 
 \item[4.] Set $G^* = G/\nu_{M_{c}}$ and use CUT or CUT \& RESOLVE to solve the FASP on on the component $\Cc(c^*)$  of arc connected cycles containing $c^*=c/\nu_{M_c}$. 
 \item[5.] Choose a new cycle $c'$ of the resulting graph and repeat 1.-4. until no such cycle exists.
 \item[6.] Use the backtracking procedure of Proposition \ref{pro:construct_essential} to compute a feedback arc set $\ee \subseteq E$ of $G$.   
\end{enumerate}
 The union  $\nu_M$ of the meta feedback vertex sets of the meta graphs can be interpreted as arcs, which are forbidden to cut in $G$. 
 The resulting feedback arc set $\ee$ will be optimal up to this obstruction ,i.e., we have 
$$ M_{c^*}^* = M_c\setminus \nu_M\,, $$ 
where $M_{c^*}^*$ denotes the meta graph of $G^*$ with respect to $c^*=c/\nu_M$.
      Hence, the quality of this heuristic can be evaluated by measuring how good $G^*$ approximates $G$. Thus, if $|\nu_M|<<|E|$ and the weight of the forbidden arcs is very high, i.e.,  
       $$ \frac{\sum_{f \in \nu_M} \omega(f)}{|\nu_M|} >> \frac{\sum_{e \in E}\omega(e)}{|E|} $$
       the arcs of $\nu_M$ will probably  not be contained in any optimal solution, yielding the correctness of Strategy \ref{S}. Additionally, the lower bounds $\mu(G,\omega)$, $\upsilon(G,\omega)$ from section \ref{Greedy} can be used to validate correctness. 
       Analogous controls can be thought of, if we choose the alternative 3b. 
       \end{strat}

\section{Discussion}
\label{Algo}
An implementation of the described algorithms is planned to be realized. Certainly, a comparison of real run times with other approaches would be of great interest. 
So far we compare our results with other theoretical approaches. Due to the immense amount of results during the last decades we restrict our discussion to publications, 
which do not restrict themselves to very tight graph classes as \emph{tournaments} \cite{Karpinski2010} or \emph{reducible flow graphs} \cite{Ramachandran:1988}. 
Finally, we suggest how the approaches of this article might be adapted to related problems. 
\subsection{The algorithms from I. Razgon and J. Chen et al.}

Note that for given graph $(G,\omega,\gamma)$ with arc weight $\omega$ and vertex weight $\gamma$ a \emph{brute force} method of solving the FASP/FVSP is given by considering every subset 
$\ee \subseteq E$  or $\nu \subseteq V$  and check whether the graphs $G\setminus \ee $, $G\setminus \nu$ are acyclic, respectively. Since due to Remark \ref{GO}, checking for acyclicity requires $\Oc(|E|^2)$ operations,
we can generate a list of all FAS's or FVS's possessing length $l_A \leq |\Pc(E)|=2^{|E|}$, $l_V\leq |\Pc(V)|=2^{|V|}$. Choosing the cheapest FAS or FVS  yields therefore a brute force algorithm solving 
the FASP/FVSP in $\Oc(|2^{|E|}|E|^2|)$, $\Oc(2^{|V|}|E|^2)$, respectively.  

The algorithm of \cite{Razgon} solves the unweighted FVSP on simple graphs in $\Oc(1.9977^{|V|}|V|^{\Oc(1)})$. 
Compared to the brute force algorithm this yields almost no improvement. Therefore, the question occured whether the parametrised version of the FVSP could be solved by an \emph{fixed parameter tractable algorithm}.
Every NP-complete problem can be solved by a fixed parameter tractable algorithm, i.e., by choosing 
$p$ as the problem size there is an algorithm with complexity $\Oc(f(p))$, where $f$ is an on the parameter $p$ exponentially depending function. Thus, the 
term fixed parameter tractable could be misleading. The precise question is whether 
there exists an algorithm with run time $\Oc(f(k)|V|^{\Oc(1)})$ computing a FVS of length less than $k$ or determing that no such set exists.
Since the FVSP is NP-complete the function $f$ will be exponentially dependent on $k$ unless $P=NP$. 
Indeed, the algorithm of \cite{Chen} solves the parametrised version of the FVSP in 
$\Oc\big(|E|^44^kk^3k!\big)$. Thus, $f(k)=k^34^{k}k!$, increases even worse than exponentially in $k$. 
Since  a small feedback length almost always correlates to small graphs or very special graphs, e.g. tree-like graphs, even improvements of the  algorithm won't be usefull 
in many applications. Therefore, the article might be seen as an purely theoretical approach answering this question. 
Indeed, to the best of our knowledge none of the algorithms  were used for an implementation of a general FVSP/FASP-SOLVER.

In contrast,  the algorithms 
\emph{CUT} and \emph{CUT \& RESOLVE} solve the FASP or FVSP on  weighted multi-digraphs in $\Oc\big(2^{m}|E|^4\log(|V|)\big)$ and \linebreak  $\Oc\big(2^{n}\Delta(G)^4|V|^4\log(|E|)\big)$. 
The parameters $m$ and $n$ fulfill $m \leq  |E|-|V| +1$,  $n \leq (\Delta(G)-1)|V|-|E| +1$ and can be computed in $\Oc(|E|^3)$, $\Oc(\Delta(G)^3|V|^3)$, respectively.
Thus, in both cases we can efficently control the run time of the exact solutions, which enables us to a priori decide whether the given instance shall be solved exactly or 
by an heuristic, e.g., Strategy \ref{S}. This crucial difference to the other approaches and the fact that Strategy \ref{S} is an heuristic on the meta level and not on the instance itsself,
makes us confident that an implementation generates a fast and accurate FASP/FVSP-SOLVER yielding a deep impact on computational an applied sciences.

\subsection{The polytop approach from C. Lucchesi et al. } 
\begin{figure}[t!]
 \centering
 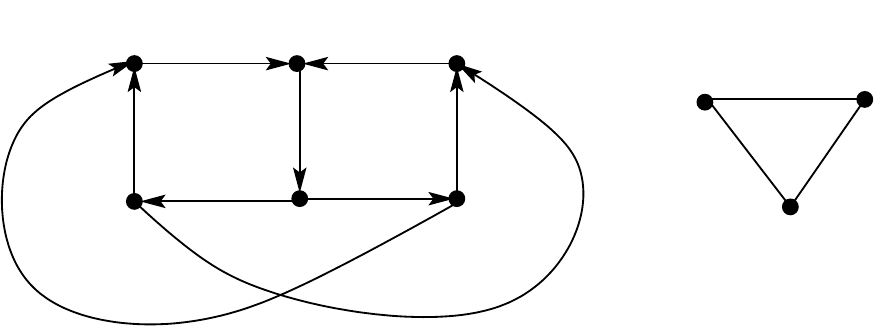
 \caption{A directed version of $K_{3,3}$.}
 \label{Ex}
\end{figure}

In \cite{Lucchesi:1978} and \cite{Groetschel:1985} a polytope of arc sets is assigned to a given graph. The FASP translates to solve a certain linear optimization over this polytope. 
In the case of \emph{planar} or more general \emph{ weakly acyclic graphs} the polytope is integral, i.e., it possesses integral corners. Since the optimum will be obtained in at least one of the corners, 
one can apply the so called  \emph{ellipsoid method for submodular functions} \cite{Groetschel}  to find the  right corner in polynomial time, 
see also \cite{marti} for further details. The approach is  certainly remarkable though it contains some weaknesses. 

The first problem is that though the algorithm runs in polynomial time the degree of the polynomial depends on a variety of parameters and cannot be estimated 
by hand a priori. Therefore, there are planar or weakly acyclic graphs, which can be solved efficently from a theoretical view point but actually a computer based implementation of the approach 
cannot ensure to meet a performance behavior applications require. 

Secondly, the class of weakly acyclic graphs is not well classified yet. Thus, if we leave the class of planar 
graphs it is hard to say whether a given graph is weakly acyclic or not. 
For instance, consider  the directed version of $K_{3,3}$ in Figure \ref{Ex}, which is known to be a weakly acyclic graph. Then it is not hard to see 
that the meta graph $M_{c}$ contains only one cycle. Thus, though $G$ is not resolvable we can efficently solve the FASP on $G$ by applying 
the algorithm \emph{CUT}. In fact all examples of weakly acyclic graphs given in \cite{Groetschel} turn out to be efficently solvable by \emph{CUT}. 
Since the techniques of this article are not sensitive to topological obstructions as planarity  we expect that there are instances of the FASP, 
which are neither planar nor weakly acyclic and even though can be solved efficently by \emph{CUT} or \emph{CUT \& RESOLVE}. 
On the other hand, by arranging cycles along several meta cycles  it is quite easy to construct a planar graph $G$ with a number $m \in \Oc(|E|)$ of linear independent meta cycles. Thus, so far 
none of the approaches can state to solve the ``larger'' instance class efficently. 

However, a deeper understanding of the meta graphs and their topology  seem to be the most relevant 
tasks for further research, which  might  enable us to classify weakly acyclic graphs 
and yield a completely new  perspective to other questions in graph theory.

\subsection{Heuristics}

In \cite{Saab} a summarization of heuristic approaches is given and several new ones are introduced. The weak point in all these approaches is that they do not provide an non-empirical 
control of the variance of the heuristical solution from the optimum. Therefore, it is impossible  to guarantee whether a solution is tight to the optimum. 
In \cite{Huang} a good lower bound of the feedback length for Eulerian graphs is given and therefore 
it would be interesting how our bound behaves on this graph class. In general, if $\ee$ denotes a feedback set the heuristic \emph{GREEDY-CUT} or 
\emph{GREEDY-CUT \& RESOLVE} proposes, then  by Proposition \ref{APX}
we  have shown that if $\omega \equiv 1$ then 
\begin{equation}\label{est}
 \max\big\{\mu(G),\upsilon(G)\big\}\leq \Omega(G,\omega)  \leq \Omega_{G}(\ee)  \leq |E|/2\,,
\end{equation}
yielding a controlled variance, as long as $\mu(G)$ can be determined  or estimated from below, see \cite{Roberts}. We conjecture that \eqref{est} improves the known estimates given by \cite{berger}. 
Furthermore, other heuristics can be
improved by the results of this article. For instance, the counter example for the Greedy approach introduced in \linebreak \cite{greedy} is resolvable and therefore 
\emph{RESOLVE \& CUT} closes this gap. 
A comparison of Strategy \ref{S}, with the approximation of \linebreak  \cite{Even:1998}, where an  approximation ratio in 
$\Oc\big(\log(\Omega)\log(\log(\Omega))\big)$ was established, might be usefull as well.

\subsection{New Approach to the Subgraph Homeomorphism Problem}
\label{HOMP}

The NP-complete \emph{directed subgraph homeomorphism problem} studied by \linebreak \cite{Fortune:1980}
is to consider  two given graphs $G=(V_G,E_G)$ and $P=(V_P,E_P)$ together
with an injective  mapping $m: V_P\lo V_G$ of vertices of $P$ into the vertices of $G$. Now the problem is given by deciding whether there exists a
injective mapping from arcs of $P$ into pairwise node disjoint elementary paths of $G$ such that an
arc $f$  with head $h$ and tail $t$ is mapped on an elementary path from $m(t)$ to
$m(h)$.
For a given  graph $G^*=(V^*,E^*)$, $u,v \in V^*$, $f=(p,q) \in E^*$ with $\deg(p),\deg(q)\geq 2$ we consider a special instance of the 
directed subgraph homeomorphism problem by setting $G=G^*\setminus f$, $P=(V_P,E_P)$ with $V_P=\{a,b,c,d\}$, $E_P=\{(a,b),(c,d)\}$ 
and $m(a)=u, m(b)=p, m(c)=q,m(d)=v$.
Thus, solving the subgraph homeomorphism problem with respect to these special instances is equivalent to decide whether $f$ is an arc of $G^*_{\el}(u,v)$. 
Hence, in addition to the polynomial time solvable instance classes
known from
\cite{Fortune:1980}, e.g. stars, also problem instances as defined above are polynomial time
solvable due to Theorem \ref{TCycle}. Potentially, our observations can be generalized in regard of this problem.


\appendix

\section{The Essential Minor}

This section is used to prove Propositions \ref{cor:ecg_equiv} and \ref{pro:construct_essential}.
To do so we recall that $\Oe(e) = \li\{c \in \Oe(G) \mi e \in \Ec(c)\re\}$ and state the following Lemmas.

\begin{lm}\label{lem:ecg_paths} 
Let $G$ be a graph, $\omega: E \lo \N^+$ be an arc weight, $\ee \in
\Sc(G,\omega)$, and $e\in E$. Then either 
\begin{equation}\label{gamma}
\ee \cap [e]_{\sim_\Gamma} = \emptyset\quad\text{or}\quad |\ee \cap [e]_{\sim_\Gamma}|=1 \,.  
\end{equation}
If in particular, $\ee \cap [e]_{\sim_\Gamma} \not = \emptyset$ then $\ee \cap [e]_{\sim_\Gamma}$ minimizes $\omega$ on $[e]_{\sim_\Gamma}$. 
\end{lm}
\begin{proof}
Let $\ee \in \Sc(G,\omega)$ and $e\in \ee$. Since every arc $f\in E$ with  with $e\sim_\Gamma f$  is connected by a branch point free path with $e$ we have that $\Oe(e)=\Oe(f)$. 
Thus, at most one  arc in $[e]_{\sim_\Gamma}$ will be cutted and this arc has to minimize $\omega$  on $[e]_{\sim_\Gamma}$.
\qed \end{proof}

\begin{lm}\label{lem:ecg_simple} 
Let $G$ be a positively weighted graph, $\ee \in \Sc(G,\omega)$, and $e\in E$.
Then either 
\begin{equation}\label{ffi}
\ee \cap [e]_{\sim_\Phi}= \emptyset\quad \text{or}\quad \ee \cap [e]_{\sim_\Phi}= [e]_{\sim_\Phi}\,. 
\end{equation}
\end{lm}
\begin{proof}
Let $\ee \in \Sc(G,\omega)$ and $e \in \ee$. Assume there is $f \in
F^+(e)\setminus \ee$  then certainly $\ee \cap F^-(e)= F^-(e)$ otherwise there
would be a two-cycle that is not cutted. Now, let $e_1,\dots,e_k \in \ee\setminus F(e)$,
$k \in \N$, be such that $\Oe(f) \cap \Oe(e_i) \not = \emptyset$ and
$\Oe\big(\{e_1,\dots,e_k\}\big)\supseteq \Oe(f)$. Since the cycles in $\Oe(e)$ and
$\Oe(f)$ differ only in a single arc, i.e., $e$ and $f$, it suffices to cut the arcs 
$F^-(e)\cup \{e_1,\dots,e_k\}$ to cut all cycles in $\Oe(e)$, i.e., 
$$\Oe\big(F^-(e)\big)\cup \Oe\big(\{e_1,\dots,e_k\}\big)\supseteq \Oe(e)\,.$$ Since
$F^-(e)\cup \{e_1,\dots,e_k\}$ is therefore a cheaper possibility than $\ee$ cutting $\Oe(F^+(e))$, this contradicts that $\ee \in \Sc(G,\omega)$ and yields the claim.
\qed \end{proof}
Now we state again Proposition \ref{cor:ecg_equiv} and deliver its proof.
\begin{pro}
Let $G=(V,E,\omega)$ be a positively weighted graph with essential minor
$(C,\omega_C)$ and let $\ee \in \Pc(E)$ and $\ee_C$ be the image of $\ee$ in $(C,\delta)$. Then 
$$\ee\in \Sc(G,\omega) \eq \ee_C \in \Sc(C,\delta)\,.$$
In particular $\Omega(G,\omega)=\Omega(C,\omega_C)$.
\end{pro}
\begin{proof}
Let $\ee \in \Pc(E)$ and $\ee_1= (\ee/_\Gamma)/_\Phi \subseteq E_1$ be the image of $\ee$ in $G_1 = (G/_\Gamma)/_\Phi$. We recall that $\omega_1= (\omega/_\Gamma)/_\Phi$ 
was defined in Definition \ref{def:ecg} and  show that 
$$\ee\in \Sc(G,\omega) \eq \ee_1 \in \Sc(G_1,\omega_1) \,\,\text{and}\,\,\Omega(G,\omega)=\Omega(G_1,\omega_1)\,.$$
Assume that $\ee \in \Sc(G,\omega)$ then by Lemmas \ref{lem:ecg_paths}, \ref{lem:ecg_simple} and the construction of $(G_1,\omega_1)$
we obtain $\Omega_{G,\omega}(\ee)= \Omega_{G_1,\omega_1}(\ee_1)$.
Thus, if  $\ee_1 \not \in  \Sc(G_1,\omega_1)$ 
then we choose 
$\alpha\in   \Sc(G_1,\omega_1)$ and a FAS $\ee' \subseteq E$ of $G$ such that the equations \ref{gamma},\ref{ffi} hold and  
$(\ee'/_\Gamma)/_\Phi =\alpha$. Consequently  $\Omega_{G,\omega}(\ee^\prime) =  \Omega_{G_1,\omega_1}(\alpha)$ and therefore due to the construction of $(G_1,\omega_1)$ we get
$$\Omega_{G,\omega}(\ee^\prime) =  \Omega_{G_1,\omega_1}(\alpha) < \Omega_{G_1,\omega_1}(\ee_1) = \Omega_{G,\omega}(\ee)\,,$$
which contradicts that $\ee\in \Sc(G,\omega)$. Thus, $\ee_1 \in   \Sc(G_1,\omega_1)$. 

Vice versa assume that $\ee \subseteq E$ is such that $\ee_1 \in\Sc(G_1,\omega_1)$. 
We claim that equations \ref{gamma},\ref{ffi} are satisfied by $\ee$. Assume the opposite then due to Lemmas \ref{lem:ecg_simple} and \ref{lem:ecg_paths}  we can delete an arc $e \in \ee$ 
or replace an arc $e \in \ee$ by an arc $f \in [e]_{\sim_\Gamma}$ with $\omega(e) > \omega(f)$. If this is not the case then we can delete all arcs $f \in F^+(e) \cap \ee$ whenever  $e$ is such that  
$\emptyset \not= F^+(e) \cap \ee \not = F^+(e)$. If $\ee^\prime$ denotes this modified set, then $\ee'$ is FAS of $G$ and in all cases 
$$ \Omega_{G_1,\omega_1}(\ee_1)> \Omega_{G_1,\omega_1}(\ee_1').$$ 
A contradiction! Hence, the equations \ref{gamma},\ref{ffi} hold for $\ee$ and therefore the construction of $(G_1,\omega_1)$ yields
$$\Omega_{G,\omega}(\ee) =  \Omega_{G_1,\omega_1}(\ee_1)=\Omega(G_1,\omega_1)\,.$$  
Thus, if $\ee \not \in  \Sc(G,\omega)$ then we choose 
$\beta \in \Sc(G,\omega)$ and obtain that $\beta_1$ is a FAS of $G_1$ with 
$\Omega_{G,\omega}(\beta)=\Omega_{G_1,\omega_1}(\beta)<\Omega(G_1,\omega_1)$, 
which is impossible. Hence $\ee \in \Sc(G,\omega)$ and the  claim follows  by iteration of these arguments. 
\qed \end{proof}

\begin{pro}
Let $G =(V,E,\omega)$ be a finite, connected, directed,  weighted
multigraph then we can construct $(C,\delta)$ in time $\Oc(|V||E|^2)$.
Furthermore, there is an algorithm with run time $\Oc(|E|^2)$ which constructs a
solution $\ee \in S(G,\omega)$ given a solution $\ee_C \in \Sc(C,\delta)$.
\end{pro}

\begin{figure}[t!]
\begin{minipage}{0.5\textwidth}
\begin{algorithm}[H]
	\KwIn{$G=(V,E,\omega)$}
	\KwOut{$G/_\Gamma$, $\kappa$}
	$\kappa(e) \leftarrow \{e\}\,, \forall e \in E$\;
	\For{$v\in V$}{
		\If{$\deg^-(v)=\deg^+(v)=1$}{
			Let $u,w\in V: (u,v),(v,w)\in E$\;
			$E\leftarrow E\cup \{(u,w)\}\setminus\{ (u,v),(v,w) \}$\;
			$V\leftarrow V\setminus \{v\}$\;
			$\omega((u,w))\leftarrow\min\{\omega((u,v)),\omega((v,w))\}$\;
			$\kappa((u,w))\leftarrow\argmin^*_{\{(u,v),(v,w)\}}\omega(\cdot)$\;
		}
        }
	\Return $(G,\omega)$, $\kappa$
	\vspace{1mm}
	\caption{$G/_\Gamma$ \label{alg:gamma}}
\end{algorithm}
\end{minipage}\quad 
\begin{minipage}[t!]{0.5\textwidth}
	\begin{algorithm}[H]
	\KwIn{$G=(V,E,\omega)$, $\kappa$}
	\KwOut{$G/_\Phi$, $\kappa$}
	\For{$e=(u,v)\in E$}{
		\For{$f\in F^+(e)\setminus \{e\}$}{
			$\omega(e)\leftarrow \omega(e)+\omega(f)$\;
			$E\leftarrow E\setminus \{f\}$\;
			$\kappa(e)\leftarrow \kappa(e)\cup\{f\}$\;
					}
        }
	\Return $G$, $\kappa$
	\vspace{1mm}
	\caption{$G/_\Phi$, $\kappa$ \label{alg:phi}}
	\vspace{14,2mm}
	\end{algorithm}
\end{minipage}
\end{figure}

\begin{proof}
The graph $G/_\Gamma$ can be computed in a single iteration over $V$, see
Algorithm \ref{alg:gamma}, where $\kappa$ is explained later. Each non branching node $v$ is removed and its two incident
arcs $e=(u,v)$ and $f=(v,w)$ are replaced by an arc $(u,w)$ with weight
$\min\{\omega(e),\omega(f)\}$. 
This is possible in time $\Oc(|V|)$ if the graph is represented as  adjacency
list where the targets of the outgoing arcs and the origins of the ingoing
arcs are stored separately. Thus, the iteration over $V$ yields a runtime of $\Oc(|V|^2)$.

To construct $G/_\sim\Phi$ we iterate over $E_\Gamma$ yielding $G_1= (G/_\Gamma)/_\Phi$, see
Algorithm \ref{alg:phi}. For each arc $e$ the weight is updated to $\sum_{e'\in
F^+(e)}\omega(e')$ and the parallel arcs $F^+(e)\setminus\{e\}$ are purged
from the graph. 
This can be realized in time $\Oc(|E|+|V|)=\Oc(|E|)$ with a counting sort
prepossessing step if the target nodes in the adjacency
list are stored such that equal targets are stored consecutively.
Because the construction of $(C,\delta)$ requires at most $|V|$ iteration steps, i.e., if $(G_K,\omega_K) = (C,\delta)$ then $K\in\Oc(|V|)$, 
the essential minor $(C,\delta)$ can be computed in time $\Oc(|V|^3 + |V||E|^2)\subseteq \Oc(|V||E|^2)$.

A simple extension of the algorithms allows to compute the information that is
necessary to compute a solution $\ee\in\Sc(G,\omega)$  once $\ee_C\in \Sc(C,\delta)$ is given.
During the application of $\Gamma$ and $\Phi$ we store the set of arcs of
$G$ that are part of a solution if the corresponding arc from $G/_\Gamma$ and
$G/_\Phi$, respectively, are in a FAS. That is, an arc that gave the minimum
weight of the two arcs in a non branching path or all parallel arcs,
respectively, see Lemma \ref{lem:ecg_paths} and Lemma \ref{lem:ecg_simple} .
In Algorithms \ref{alg:gamma},\ref{alg:phi} this is realized by $\kappa$ which can be considered
as $\kappa: E_C\lo \Pc(E)$. The mapping is initialized as 
$\kappa(e)\leftarrow \{e\}$. 
Storing the arcs as linked list allows to update $\kappa$ in linear time, i.e., 
the asymptotic run time of Algorithms \ref{alg:phi},\ref{alg:gamma} remains unchanged. 
Note that, $\argmin^*$ returns only one arc in the case of equality. 
Now, replacing each arc $e\in\ee_C$ by $\kappa(e)$ yields $\ee$, which due to Proposition \ref{cor:ecg_equiv} 
is a solution for the FASP on $(G,\omega)$.
Thus, the  replacement can be realized in time $\Oc(|E|^2)$. 
\qed \end{proof}

\begin{remark}
Algorithm \ref{alg:gamma} may be extended to generate all solutions of the FASP for $G$
given all solutions for the FASP on $C$. Therefore the equal weight alternatives in a non
branching path need to be stored. The generation of the combinations of the
alternatives of different paths yields all solutions. Certainly, then the run time 
depends exponentially on the number of possible combinations of alternatives. 
\end{remark}

\begin{acknowledgements}
I want to thank Matthias Bernt for many fruitful discussions and his support for formalizing some algorithms. Moreover, a  heartful thank you goes to Peter F. Stadler
for the nice time at his bioinformatics institute in Leipzig.
\end{acknowledgements}


\end{document}